\documentclass[12pt,reqno]{amsart}
\usepackage{mathrsfs}               % \mathscr
\usepackage{amssymb}
\usepackage{graphicx, amsmath}
\usepackage{amsfonts}
\usepackage{dsfont}
\usepackage{amssymb}
\usepackage{fancyhdr}
\usepackage{a4wide}
\usepackage{xcolor}
\usepackage{enumerate} 
\usepackage[colorlinks]{hyperref}
\usepackage{mathtools}

\newcommand{\form}[1]{q_0[#1]}

\newcommand{\SP}[2]{\big\langle #1,#2 \big\rangle} %S. Product BIG
\newcommand{\sps}[2]{\langle #1,#2 \rangle} %Scalar Product
\newcommand{\bx}{{\bf x}}
\renewcommand{\d}{\textrm{d}}

\newcommand{\N}{\mathbb{N}} % Natural numbers
\newcommand{\Z}{\mathbb{Z}} % Integer numbers
\newcommand{\R}{\mathbb{R}} % Real numbers
\newcommand{\U}{\mathcal{U}}
\newcommand{\calD}{\mathcal{D}}
\newcommand{\calQ}{\mathcal{Q}}

\newcommand{\F}{\mathcal{F}}
\renewcommand{\S}{\mathbb{S}}
\renewcommand{\H}{\mathcal{H}}
\newcommand{\calC}{\mathcal{C}}
\newcommand{\calU}{\mathcal{U}}
\newcommand{\re}{\mathrm{Re}}

\newcommand{\veps}{\varepsilon}
\newcommand{\wti}{\widetilde}

\newcommand{\n}[1]{\left\| #1 \right\|}

\newcommand{\ab}[1]{\left| #1 \right|}
\newcommand{\norm}[1]{\left\| #1 \right\|}
\newcommand{\abs}[1]{\left| #1 \right|}
\newcommand{\supp}{\textnormal{supp}}
\newcommand{\<}{\left\langle}
\renewcommand{\>}{\right\rangle}
\newcommand{\vp}{\varphi}
\newcommand{\ve}{\varepsilon}
\renewcommand{\t}[1]{\textnormal{#1}}

\renewcommand{\(}{\left(}
  \renewcommand{\)}{\right)}
\newcommand{\bl}{}

\newtheorem{condition}{Condition}
\newtheorem{theorem}{Theorem}[section]
\newtheorem{proposition}[theorem]{Proposition}
\newtheorem{lemma}[theorem]{Lemma}
\theoremstyle{remark}
\newtheorem{remark}[theorem]{Remark}
\newtheorem{remarks}[theorem]{Remarks}

\theoremstyle{definition}

\newcounter{listi}
\newenvironment{remarklist}{\begin{list}{{\rm(\roman{listi})}}{%
\setlength{\topsep}{0mm}\setlength{\parsep}{1mm}\setlength{\itemsep}{0mm}%
\setlength{\labelwidth}{1.3em}\setlength{\leftmargin}{1.2em}\usecounter{listi}%
}}{\end{list}}

\renewcommand{\le}{\leqslant}
\renewcommand{\leq}{\leqslant}
\renewcommand{\ge}{\geqslant}
\renewcommand{\geq}{\geqslant}     

\title[Asymptotic dynamics of magnetic quantum systems]{On the
  asymptotic dynamics of 2-D magnetic quantum systems}
\author[E. C\'ardenas]{Esteban C\'ardenas}

\address[E. C\'ardenas]{Department of Mathematics,
	University of Texas at Austin,
	2515 Speedway,
	Austin TX, 78712, USA}
	\email{eacardenas@utexas.edu}
\author[D. Hundertmark]{Dirk Hundertmark}
\address[D. Hundertmark]{Department of Mathematics, Institute for Analysis, Karlsruhe Institute of Technology, 76131 Karlsruhe, Germany, and Department of Mathematics, Altgeld Hall, University of Illinois at Urbana-Champaign, 1409 W. Green Street, Urbana, IL 61801, USA}
\email{dirk.hundertmark@kit.edu}

\author[E. Stockmeyer]{Edgardo Stockmeyer}
\address[E. Stockmeyer]{Instituto de F\'isica, Pontificia Universidad Cat\'olica de Chile, Vicu\~na Mackenna 4860, Santiago 7820436, Chile}
\email{stock@fis.puc.cl}

\author[S. Wugalter]{Semjon Wugalter}
\address[S. Wugalter]{Department of Mathematics,
Institute for Analysis,	Karlsruhe Institute of Technology,
	76131 Karlsruhe, Germany}
	\email{semjon.wugalter@kit.edu}

\begin{document}
\thanks{\copyright 2020 by the authors. Faithful reproduction of this article,
       in its entirety, by any means is permitted for non-commercial purposes}
%\keywords{?????}
%\subjclass[2020]{35Q55, 35Q60, 35A01, 35B35, 35B30}
%\date{\today, version \jobname }
\begin{abstract}
  In this work we provide results on the long time localisation in space
  (dynamical localisation) of certain two-dimensional magnetic quantum
  systems. The underlying Hamiltonian may have the form $H=H_0+W$,
  where $H_0$ is rotationally symmetric, has dense point spectrum,
  and $W$ is a perturbation that breaks the rotational symmetry. In the latter case, we also give estimates for the growth of the angular momentum operator in time.

\end{abstract}

\maketitle
{\hypersetup{linkcolor=black}
\setcounter{tocdepth}{1}
\tableofcontents}

\section{Introduction}
Consider a charged quantum particle subject to a time independent electro-magnetic field 
profile. The system may be described through a self-adjoint operator $H$ with 
domain $\calD(H)$ in a Hilbert space $\mathcal{H}$. If we  assume that the system is initially in a 
state $  \varphi\equiv \varphi(0)\in \calD(H)$ then, according to the Schr\"odinger equation,   the state of the system  at time $t$, $\varphi(t)$,  is given by $e^{-iHt/\hbar}\varphi$ (here $\hbar>0$ is Plank's constant devided by $2\pi$). A fundamental question is whether the 
system remains localized for long times and, if not, what  the speed of the wavepackage spreading 
is in terms of the electromagnetic field profile.

This
phenomena can be investigated, for instance, by looking at the long-time behavior of the expected radius of the state
\begin{align}\label{radius}
\sps{\varphi(t)}{\abs{\bx} \varphi(t)}\,, \quad \mbox{for} \quad t\gg 1\,.
\end{align}
 This, in turn, can sometimes be estimated if one has
information on the spectral quality of the underlying Hamiltonian (see,
e.g., \cite{Oliveira2009,Teschl1999} and \cite{Last1996}).  Let us assume 
that the initial state belongs to  a finite energy region $I\subset \R$ (with $\abs{I}<\infty$) of $\mathcal{H}$, i.e., 
$\varphi=E_I(H) \varphi$, with $E_I(H)$ being the spectral projection of $H$ on $I$.
Then, one can
easily check that if the spectrum of the Hamiltonian is a discrete
set, then the system remains localized in the sense that 
 $$
 \sup_{t\in\R}\sps{\varphi(t)}{\abs{\bx} \varphi(t)}\le C \norm{\varphi}^2\,,
 $$
for some constant $C>0$. Moreover, in dimension one, it is known that if the spectrum is absolutely continuous the wavefunction spreading is ballistic  in time average. More precisely, there  is a constant $c>0$ such that
\begin{align*}
\frac{1}{T} \int_{0}^T  \sps{\varphi(t)}{\abs{\bx} \varphi(t)} \ge c T\,, \quad T>1\,.
\end{align*}
(See \cite{Last1996} for this and more general results of this type).
However, if the spectrum is
dense point or singular continuous there is very little one can
say a priori (see \cite{RJLS2}).  Indeed, for $H$ having point spectrum, it is only known in general that the
system exhibits a sub-ballistic dynamical behaviour \cite{Simon1990},
i.e.,
\begin{align*}
\lim_{t\to\infty} \sps{\varphi(t)}{\abs{\bx} \varphi(t)}/t\to 0\,.
\end{align*}
Moreover, there are examples  of Hamiltonians with pure point
spectrum where the spreading rate is arbitrarily close to 
ballistic \cite{RJLS2}, i.e., for any $\varepsilon>0$
\begin{align*}
\limsup_{t\to\infty}  \sps{\varphi(t)}{\abs{\bx} \varphi(t)} /t^{1-\varepsilon}\to \infty\,,
\end{align*}
for a large class of initial data $\varphi$.

In the work at hand, we shed some more light on this problem for the cases,  (a) when $H=H_0$ is the two-dimensional magnetic Schr\"odinger operator with a radially symmetric magnetic field $B$ and has dense point spectrum and (b) when $H=H_0+W$,  with $H_0$ as before and   $W$ being 
an electric perturbation, smooth in the angular variable,   that decays at infinity.  Our conditions include the cases when 
\begin{align*}
\int_0^r B(s)s ds =\lambda r^{\sigma}\,, \quad \lambda>0, \sigma\ge 1\,.
\end{align*}
Using the same arguments as in \cite{miller1980quantum} one can easily show that for $\sigma \in (1,2)$ the spectrum of $H_0$ is dense pure point. Moreover, if $\sigma=1$ there is a {\it mobility edge} at energy $\lambda^2$, i.e., the spectrum is dense pure point on $[0,\lambda^2)$ and purely absolutely continuous on $(\lambda^2, \infty)$. It follows directly from Theorem \ref{thm 1} below  that  when $\sigma\in (1,2)$ the dynamics generated by $H_0$ is localized in time, provided the initial data is sufficiently smooth. Moreover, we show an anlogous result for the case $\sigma=1$,  whenever   $\varphi=E_{[0,\lambda^2)}(H) \varphi$ (see Section~\ref{med}). Similar results have been obtained for Dirac operators in 
	\cite{Barbaroux}. 
 
 The problem for $\sigma \in (1,2)$ becomes much more delicate if we turn on the electric perturbation $W$. In this case we do not even know the quality of the spectrum.  Indeed, through the perturbation,
 certain spectral subspaces may cease to be pure point and continuous
 spectrum (presumably singular) may appear (see e.g., \cite{RMS1994} and \cite{Oliveira2009}).  In this case, we provide estimates on the wave package spreading in terms of the decay rate of $W$. In particular, we show  that if $W$ decays exponentially fast, then the expected radius of the system grows at most logarithmically fast in time. Moreover, if $W=\mathcal{O}(1/\abs{\bx}^p)$ for some $p>2\sigma$, then $\sps{\varphi(t)}{\abs{\bx} \varphi(t)} $ grows at most as $t^\theta$ with $\theta=(p-\sigma)^{-1}<1$ (see Theorem~\ref{coro 1}, below).
 
 In order to prove that, we show on the one hand, that one can control the growth of the radius \eqref{radius}  in terms of the expected, time-dependent, angular momentum. (This is the actual content of   Theorem~\ref{thm 1}.) On the other hand, in Theorem~\ref{thm 2}, we provide estimates on the growth in time  of the angular momentum operator in terms of the decay rate of $W$.  As an essential tool, we use  certain novel  tunnelling estimates (see Theorem~\ref{thm-tunnelling}) which are in turn 
derived from fairly general exponential decay estimates for the spectral projections $E_{I}(H)$ given in Theorm~\ref{thm 3}.

 The above discussion roughly summarizes our main results. We emphasize that we are not aware of
other localisation bounds of this type  in  such situations (perturbed dense point spectrum) for
deterministic systems. Notice, however, that   the subject is frequently addressed  in the
realm of random Schr\"odinger operators. Although in these cases  the randomness of the potential plays a fundamental role in the proofs of localisation (see, e.g.,
\cite{hundertmark2008short,aizenman2015random}).

{\it This  paper is organized as follows}: In the rest of this section we describe precisely the model   and state most of our main results. We show theorems~\ref{thm 1} and \ref{thm 2} in Section~\ref{section thm 2}. In Section~\ref{edcay} we state and prove the exponential decay estimates for  the spectral projections Theorem~\ref{thm 3}. Finally,  in Section~\ref{aplications}, we apply the latter theorem to the model at hand and  show   the tunnelling estimates stated in Theorem~\ref{thm-tunnelling}.  
\subsection{The model and main results}
 Let us introduce the Hamiltonian $H_0$ of a quantum particle moving in
$\R^2$ that is interacting with a magnetic field $\textbf{B}$ pointing
perpendicularly to the plane. We denote by
$\textbf{A}=(A_1,A_2):\R^2\rightarrow \R^2$ a magnetic vector
potential associated to the magnetic field through
$\textbf{B}=(\partial_1 A_2-\partial_2 A_1) \hat{x}_3$.  
Throughout this work we use units such that $\hbar =2m=1$, where $m$ is the mass of the particle. 
For $\textbf{A}\in L_{\textrm{loc}}^2(\R^2,
\R^2)$ we define the sesquilinear form 
\begin{align}
  \label{eq:22}
  q_0(\varphi,\psi)=\int_{\R^2}
  \overline{(-i\nabla-\textbf{A}(\bx))\varphi(\bx) } (-i\nabla-\textbf{A}(\bx))\psi(\bx) \d \bx\,,
\qquad\varphi,\psi\in  \calD(q_0)\,,
\end{align}
with domain 
\begin{align}
  \label{eq:23}
  \calD(q_0)=  \left\{ 
\varphi\in L^2(\R^2)\,|\,
 (-i\partial_j-A_j) \varphi \in L^2(\R^2)\,, j\in \{1,2\}
\right\}\,.
\end{align}
%{ For quadratic forms we use the standard notation $q_0[\varphi]=q_0(\varphi,\varphi)$.}

It is well known \cite{Simon1979} that $C_0^\infty(\R^2)$ is a form core for $
q_0$.  
We denote by $H_0$ the self-adjoint operator {corresponding} to $q_0$ and by
$\calD(H_0)\subset  \calD(q_0)$ its domain.

We are interested in the particular case in which $H_0$ describes the
dynamics of a particle in a rotationally symmetric magnetic field
$\textbf{B}(\bx) =B(|\bx|) \hat{x}_3$. We choose the Poincar\'e gauge where (here
$r:=\ab{\bx}$, as usual)
\begin{align}\label{rotgauge}
{\bf A}(\t{\textbf{x}})=A(r)\hat{\theta} = \frac{ 
\Phi(r)}{r^2} 
\begin{pmatrix} 
-x_2 \\ x_1
\end{pmatrix}
\quad
\hbox{with}\quad
 \Phi(r)= {A(r)}{r}=\int_{0}^{r} B(s)s \, \t{d} s\,.
\end{align}
We will show that this choice of vector potential is locally square integrable, whenever the magnetic field is, see Lemma \ref{lem-square-integrable} in the appendix.  
Notice that $\Phi(r)$ is, up to factor of $2\pi$, the magnetic flux through a disc of radius $r>0$ centered at the origin.
%, as $\Phi\in L^2_{\rm loc} (\R^+,\d r/r )$ where
%\begin{align}\label{rotgauge}
%\Phi(r)= rA(r)\,, \quad r\ge 0\,.
%\end{align}

One can show, see the discussion in Appendix \ref{app-polar coordinates}, 
that the quadratic form $q_0$ corresponding to $H_0$ is given by 
\begin{equation}\label{eq-q_0-rep-1}
	q_0(\vp,\psi) 
	  = \SP{\partial_r\vp}{\partial_r\psi} 
	     + \SP{\frac{1}{r}(\Phi-L)\vp}{\frac{1}{r}(\Phi-L)\psi} 
\end{equation}
for all $\vp,\psi\in\calD(q_0)$ since the magnetic flux $\Phi$ is radial. 
Here $\partial_r= \tfrac{x}{|x|}\cdot\nabla$ is the radial derivative and $L= -i (x_1 \partial_2 - x_2 \partial_1)$ is the generator of rotations in $\R^2$. 

Identifying the underlying Hilbert space $L^2(\R^2)$  with  
$\H:=L^2(\R^+
\times \S^1, r\d r\, \d\theta)$ through the
transformation
\begin{equation}
  \label{eq:8}
  \calU: L^2(\R^2)\to \H\,,
\end{equation}
with $\psi\mapsto \U\psi= \tilde{\psi}$, where
$\tilde{\psi} (r,\theta)= \psi(r\cos\theta,r\sin\theta)$, 
we define the self-adjoint  angular momentum operator $J=\U L\U^{-1}$. It is easy to see that  
\begin{align}
  \label{eq:6}
  J\widetilde{\varphi}:=-i\frac{\partial}{\partial\theta} \widetilde{\varphi}
\end{align}
when $\widetilde{\varphi}=\U\varphi$. 
In this coordinates we have, for any $\varphi,\psi\in \calD(q_0)=\calQ(H_0)$,
\begin{align}
  \label{eq:24}
  q_0(\varphi,\psi)=
\SP{\partial_r \tilde{\varphi}}{\partial_r    \tilde{\psi}}_\H
+\SP{\tfrac{1}{r} (\Phi-J)\tilde{\varphi}}{\tfrac{1}{r} (\Phi-J) \tilde{\psi}}_\H\,,
\end{align}
where $\tilde{\varphi}=\calU\varphi$ and
$\tilde{\psi}=\calU\psi$ with the unitary $\U:L^2(\R^2)\to \H$ given in 
\eqref{eq:8}.

 Notice that the spectrum of $J$ coincides with the set of integers $\Z$. 
We define
$P_j$ to be the eigenprojection onto the subspace of $\H$ with
fixed angular momentum $j\in \Z$. It is well known that the family
$(P_j)_{j\in \Z}$ gives a complete decomposition of the Hilbert space
$\H$ into subspaces which diagonalize  $q_0$: Expanding in Fourier series we have
$\U\vp (r,\theta)=(2\pi)^{-1/2}\sum_{j\in\Z} \vp_j(r) e^{ij\theta}$ and similarly for $\psi$ and using \eqref{eq:24} gives 
\begin{equation}
 \begin{split}
  \label{eq:12}
 q_0(\varphi,\psi)
  &=\sum_{j\in \Z} \left(\SP{\partial_r \varphi_j}{\partial_r    {\psi}_j}_{L^2(\R_+,rdr)}
    +\SP{\tfrac{1}{r}(\Phi-j) {\varphi}_j}{\tfrac{1}{r}(\Phi-j)\psi_j}_{L^2(\R_+,rdr)}\right) \\
  & =\sum_{j\in \Z} \left(\SP{\partial_r \varphi_j}{\partial_r    {\psi}_j}_{L^2(\R_+,rdr)}
    +\SP{\varphi_j}{V_j\psi_j}_{L^2(\R_+,rdr)}\right) \\
  &\eqqcolon \SP{\partial_r \varphi}{\partial_r    {\psi}}
        +\SP{\varphi}{V\psi}
 \end{split}
\end{equation}
where $V_j(r) := \tfrac{1}{r^2}(\Phi(r)-j)^2$ is the {\it effective potential}. We will frequently identify $\varphi_j$ with $P_j\calU\varphi$ and similarly for $\psi_j$.

We consider  electric perturbations of $H_0$  through a potential $W$
which is not necessarily rotationally symmetric but satisfies the following smoothness condition in the angular
variable.
\begin{condition} \label{con3}
  There are constants $a>0$,  $0<\zeta\le 1$, and a function
  {\bl $v\in L^2(\R^2)+L^\infty(\R^2)$} such that for all $j\in \Z$ and
  almost every $r>0$ 
  \begin{align}\label{asbe}
    \ab{\widehat{W}(r,j)}\le b(r) e^{-a\ab{j}^\zeta}\,,
  \end{align}
where $b(\ab{\bx})=v(\bx)$ for $\bx \in \R^2$ and, for $j\in \Z$,
\begin{align*}
  \widehat{W}(r,j)\coloneqq\frac{1}{\sqrt{2\pi}} \int_0^{2\pi} W(r,\theta)
  e^{-ij\theta} \t{d} \theta\,,\quad \mbox{for a.e.} \quad r>0\,.
\end{align*}
is the Fourier transform of the potential $W$ in the angular variable. 
\end{condition}

\begin{remarks}
\begin{remarklist}\vspace{1mm}
	\item The above condition is an analyticity condition in terms of the Gevrey scale in the angular variable $\theta$. In particular, if $\zeta=1$, the  above condition is precisely the analyticity $W$ in $\theta$, for almost every $r>0$,  which is also clear from the familiar Paley--Wiener theorem.  
	    Such an analyticity  condition is not unusual, see \cite{laszlo} 
	    and \cite{Shu-gaussian}.
\item \label{rel-bdd} 
%  In view of  \cite[Satz 17.7]{WII}, we know that $v$ is infinitesimally $-\Delta$-bounded. Moreover, by Lemma~\ref{w-and-ef}, this is also true for $W$ instead of $v$. }
%In addition, this also holds  in the sense of quadratic forms.
In view of the diamagnetic inequality, see Theorem 2.4 and Theorem 2.5
from \cite{AHSI}, we see that $W$ is  infinitesimally $H_0$-bounded in
the operator and form sense. In particular, we have that for any
$\varepsilon>0$ there exist constants $C(\varepsilon)>0$ such that
\begin{equation}
\label{eq:3}
  \norm{W\varphi}^2\le \varepsilon\norm{H_0
  \varphi}^2+C(\varepsilon)\norm{\varphi}^2\,,\quad
\varphi\in \mathcal{D}(H_0)\,
\end{equation}
and 
\begin{equation}
\label{eq:5}
  \ab{\int_{\R^2}W(\bx)\ab{\varphi(\bx)}^2 \d \bx }\le \varepsilon \,\form{\varphi}+C(\varepsilon)\norm{\varphi}^2\,,\quad
\varphi\in \calD(q_0)\,.
\end{equation}
\item Notice that the diamagnetic inequality  (form bounded with respect to the nonmagnetic kinetic energy implies   the same for the magnetic kinetic energy) in this gauge, is an easy consequence of our analysis of the magnetic quadratic form, see Appendix \ref{app-polar coordinates}: Lemma \ref{lem radial energy boundedness} and Remark \ref{rem:diamagnetic}. 
\end{remarklist}
\end{remarks}
In the work at hand we study the dynamics of a quantum particle
governed by the Hamiltonian
\begin{align}
  \label{eq:4}
  H\varphi:=H_0\varphi+W\varphi\,,\qquad \varphi\in \calD(H)=\calD(H_0)\,.
\end{align}
{  In view of Remark \ref{rel-bdd}}  and the Kato-Rellich theorem the operator $H$
is bounded from below and self-adjoint. 

For $\vp,\psi$ in the form domain of $W$, $\calQ(W)$, one checks that as quadratic forms 
\begin{align}
	\SP{\vp}{W\psi} = \sum_{j,k\in\Z} \SP{\vp_j}{\widehat{W}(\cdot,j-k)\psi_k}_{L^2(\R_+,rdr)}
\end{align}
and thus the quadratic form of the magnetic Schr\"odinger operator $H=H_0+W$ is given by 
\begin{equation}\label{def-magn form}
  \begin{split}
	q(\vp,\psi)&= \SP{\vp}{H\psi} \\
	 &= \sum_{j\in \Z} \left(
	         \SP{\partial_r \varphi_j}{\partial_r    {\psi}_j}_{L^2(\R_+,rdr)}
             +\SP{\varphi_j}{V_j\psi_j}_{L^2(\R_+,rdr)}
         \right) \\
     &\phantom{==}       +  \sum_{j,k\in\Z} \SP{\vp_j}{\widehat{W}(\cdot,j-k)\psi_k}_{L^2(\R_+,rdr)} \\
     &=  \SP{\partial_r \varphi}{\partial_r    {\psi}}
        +\SP{\varphi}{V\psi} + \SP{\varphi}{W\psi}\,.
  \end{split}	
\end{equation}
 see also the discussion in Appendix \ref{app-polar coordinates} for more details.

In order to state our results in a concise way we will work separately with the following two conditions on the magnetic flux  $\Phi $ given in \eqref{rotgauge}. 
\begin{condition}\label{con1}
	Let $\Phi \in L^2_{\rm loc}(\R^+, \d r/r)$ such that there are  constants 
	$\lambda_+<\infty $ and $\sigma_+ > 1$ such that
	\begin{align}
	\label{int-con}
	|\Phi(r)|\le \lambda_+(1+ r^{\sigma_+})\,, 
	\end{align} 
	for all $r>0$. 
\end{condition}
\begin{condition}\label{con1b}
	Let $\Phi \in L^2_{\rm loc}(\R^+, \d r/r)$ be such that there are 
	constants $ r_0>1$, $\lambda_->0 $ and 
	$\sigma_-> 1$ such that
	\begin{align}
	\label{ext-con}
	|\Phi(r)|\ge \lambda_-\, r^{\sigma_-}\quad r\ge r_0
	\,.
	\end{align} 
\end{condition}
\begin{remark}\label{rk1}
It is interesting  to consider  the model case where 
the magnetic field is bounded and asymptotically decays to zero as   
	\begin{align}
	\label{eqn:3}
	B(r)\sim  r^{-\alpha},\qquad \text{for } r\to\infty\,,
	\end{align}
	for some $\alpha< 1$. Here conditions \ref{con1} and \ref{con1b} are fulfilled with  $\sigma_+=\sigma_-=2-\alpha$.   

\end{remark}
Our results concern the time evolution of a state $\varphi$ with
energy on a bounded interval $I\subset \R$. In order to state them we
introduce some notation that will be used throughout this work.
{  Let $e_0:={\rm inf}\, {\rm spec}(H)$, be the infimum of the spectrum of $H$}  and let $E_0 \in (e_0, \infty)$ be a
fixed constant. We set $$ I:=[e_0,E_0]$$ and denote by $E_I(H)$ the
spectral projection of $H$ onto the interval $I$.
Let $U(t)\coloneqq e^{-itH}$ be the time evolution operator associated to $H$.
For any initial state $\vp\in L^2(\R^2)$ we denote by $\vp(t):=U(t)\vp$ the
state of the system at time $t$.

We are now ready to state our main results.
{
\begin{theorem}\label{thm-tunnelling}
	Assume that Condition \ref{con3}  is satisfied.
	\begin{enumerate}[i{\rm)}]
		\item {\rm(}Interior tunnelling estimates{\rm)}. 
Under Condition \ref{con1} there exist constants $c_+\in (0,1]$ and $\delta_+>0$ such that 
	\begin{align}
	\label{lemma 1.1 a}
	\sum_{j\in\Z} e^{\delta_+ |j|^\zeta}\| \mathds{1}_{ [0,c_+ |j|^{\zeta/\sigma_+} ]  }(\ab{\bx}) P_j E_I(H)\|^2 < \infty\,.
	\end{align}
\item (Exterior tunnelling estimates). Under Condition \ref{con1b}  
there exist constants $c_-\ge r_0$ and $\delta_->0$ such that 
	\begin{align} 
	\label{lemma 1.1 b}
	\sum_{j\in\Z}  \| \mathds{1}_{ [c_- |j|^{\zeta /\sigma_-} ,\infty ) }(\ab{\bx}) \, e^{\delta_- |\t{\textbf{x}}|^{\zeta\sigma_-} } \, P_j E_I(H)\|^2 < \infty\,.
	\end{align}
	Here $\mathds{1}_A$ stands for the indicator function of the set $A$ and {   $\sigma_-$ and $\sigma_+$ are  the parameters given in condition \ref{con1} and \ref{con1b}.}
	\end{enumerate}
\end{theorem}
}
\begin{remarks}
\begin{remarklist}
  \item 
   The interior and exterior tunnelling bounds above show \emph{strong decay} of the spectral projection $E_I(H)$ for finite energy intervals $I$ into the classically forbidden region. They are derived from the exponential decay of the energy projections described in Section~\ref{edcay}. Remarkably, these bounds are valid in a regime, where the unperturbed operator $H_0$ has dense point spectrum. 

\item\label{tunn estimate}
 The bound  \eqref{lemma 1.1 a}  also immediately implies that there 
 exists  a positive constant $C $ such that 
	\begin{equation}
	\n{ \mathds{1}_{ [0,c_+ |j|^{\zeta/\sigma_+} ]  }(\ab{\bx})\, P_j E_I(H) } \leq C\, e^{-\delta_+ |j|^\zeta}\,,  \qquad j\in\Z\,.
	\end{equation}
 \item \label{tunn estimate2}
	Consider the example of Remark \ref{rk1}. 	Theorem \ref{thm-tunnelling} indicates  that a wave-function with energies in the interval $I$ and angular momentum $j\in\Z$, $P_jE_I(H)\varphi$, is essentially localized in the annulus between  $c_+ |j|^{1/(2-\alpha)}$ and $c_- |j|^{1/(2-\alpha)}$. This is in fact the scale where the classically allowed region (see \eqref{class allowed} below) for $P_jE_I(H)\varphi$ is located. For further details 
	see Section \ref{aplications}, where the proof of Theorem \ref{thm-tunnelling} is given.
\end{remarklist}
\end{remarks}

 Our next theorem states
that the expectation of $\ab{\bx}$ in time, is dominated by the
expectation of the angular momentum operator in time to certain
power. This power  depends on the behavior of the  magnetic flux far
from the origin, see Condition \ref{con1b}. 
\begin{theorem}[{$J(t) $ controls $x(t)$}]
\label{thm 1}
Let $H$ be the Hamiltonian defined in \eqref{eq:4} and assume that conditions \ref{con3} and
\ref{con1b} are satisfied. Then, for any $\nu>0$, there exists a constant $C>0$ such
that for all {  initial states $\vp\in E_{I}(H) L^2(\R^2)$} 
\begin{equation}
\SP{\varphi(t)}{|\bx|^{\nu }\vp(t)} \leq C \( \n{\vp}^2  +    \SP{\vp(t)}{|J|^{\zeta\nu/{ \sigma_-}} \vp(t)} \)\,, \quad \textnormal{for}\quad t\in \R\,.
\end{equation}
\end{theorem}
\begin{remarks}
\begin{remarklist}
\item \label{rt1}
   Assume further that $W$ is rotationally symmetric. Then, the time evolution $U(t)$
   commutes with $\ab{J}$ (i.e., angular momentum is conserved), hence Theorem \ref{thm 1} implies dynamical
   localization for any $\varphi\in \calD(|J|^{\zeta \nu/(2\sigma_-)})$, i.e., 
\begin{align*}
  {\rm sup}_{t\ge 0} \SP{\varphi(t)}{|\bx|^{\nu }\vp(t)}<\infty\,.
\end{align*}
\item \label{rt2}
The main point of Theorem \ref{thm 1} is the dynamical estimate when
$E_0\ge 0$. If $E_0<0$ then one can easily prove dynamical
localization since the spectrum of $H$ below zero is discrete.
	
\end{remarklist}
\end{remarks}
%%%%%%%%%%%%%%%%%%%%%%%%%%%%%
In order to formulate the next theorem we define the symmetric and
non-symmetric parts of the potential $W$ by writting
\begin{align*}
  W=W_{\rm s}+W_{\rm ns}\,,
\end{align*}
where $W_{\rm s}$ is the radially  symmetric part of $W$ given, for
almost all $r>0$, by
\begin{align*}
  W_{\rm s}(r):=\frac{1}{{2\pi}} \int_0^{2\pi} W(r,\theta) \t{d} \theta\,.
\end{align*}
If we assume some further  decay in space for $W_{\rm ns}$, 
we  obtain bounds for the expectation  of the angular momentum in
time. We state our results for two different classes of decay {  of $W_{\rm ns}$}.
\begin{theorem}[{Bounds on $J(t)$}]
	\label{thm 2}
Assume that conditions \ref{con3} and
\ref{con1} are satisfied.
\begin{enumerate}
\item[{\rm(}i{\rm)}] Suppose that for $p>\sigma_+/\zeta$   
\begin{equation} \label{powerlike1}
W_{\rm ns} (\t{\textbf{x}})  =  \mathcal{O}\(\frac{1}{|\bx|^p}\), \qquad |\bx|\rightarrow \infty\,.
\end{equation}
Then, for any $0<\beta <(\zeta p-\sigma_+)/\sigma_+$ there exists $C>0$ such that for all {  initial states $\vp \in  E_{I}(H) L^2(\R^2)$} we have 
\begin{equation}
\||J|^{\beta/2} \vp(t) \|^2 \leq \| |J|^{\beta/2} \vp   \|^2 + C\, t^{ \gamma\beta   }\n{\vp}^2\,, \qquad \qquad t>1\,,
\end{equation}
where $\gamma = \frac{\sigma_+}{\zeta p-\sigma_+}$.
\item[{\rm(}ii{\rm)}] Suppose that there exists $\mu>0$ and $s>0$ such that
\begin{equation}\label{gaussian}
W_{\rm ns}(\bx) = \mathcal{O} \big(\exp(-\mu |\bx|^{s})  \big), \qquad  |\bx|\rightarrow \infty\,.
\end{equation}
Then,  for any $\beta >0$,  there exists $C >0$ such that for all {  initial states $\vp\in E_{I}(H) L^2(\R^2)$} we have
\begin{equation} 
\||J|^{\beta/2} \vp(t) \|^2 \leq \| |J|^{\beta/2} \vp   \|^2  + C \,
\big( \ln (t) ^{\theta \beta } +1 \big)   \n{\vp}^2\,, \qquad \qquad t>1\,,
\end{equation} 
where $\theta = 1/{\rm min}\{\zeta,\zeta s / { \sigma_+}\}$.
\end{enumerate}
\end{theorem}

The proofs of 
  	Theorem~\ref{thm 1} and Theorem~\ref{thm 2} are given 
  	in Section \ref{section thm 2}. Let us emphasize that 
  	 while Theorem~\ref{thm 1} uses Condition \ref{con1b} through the exterior  tunnelling estimate \eqref{lemma 1.1 b}, Theorem~\ref{thm 2} requires   Condition \ref{con1} 
  	in order to apply   \eqref{lemma 1.1 a}.
  The next result is a direct combination of  Theorem~\ref{thm 1} and
Theorem~\ref{thm 2}.

\begin{theorem}[{Bounds on $x(t)$}]
	\label{coro 1}
 Assume that conditions \ref{con3} and
 \ref{con1}, and \ref{con1b} are satisfied with $0<\zeta\le 1$, 
 $1< \sigma_-\le \sigma_+$. Then 
\begin{enumerate}
  \item[(i)]  Assume that  
     $ W_{\rm ns} (\t{\textbf{x}})  =  \mathcal{O}(\tfrac{1}{|\bx|^p}),$ as $|\bx|\rightarrow \infty$,  for some $p>\sigma_+/\zeta$.  
     Then, for any $0< \nu <   \frac{\sigma_-}{\zeta}( \frac{\zeta p - \sigma_+}{\sigma_+})$, 
     there exists $C>0$ such that, for all $\vp\in E_{I}(H) L^2(\R^2)$, we have
\begin{equation}
\SP{\vp(t)} { |\bx|^{\nu} \vp(t)} 
\le C \left(   \SP{\vp}{|J|^{ \zeta \nu/\sigma_-} \vp}    +  t^{ \, \varepsilon_p  }   \n{\vp}^2     \right) \,, \qquad t>1\, .
\end{equation}
  where $\varepsilon_p =  \frac{\zeta }{\sigma_-} (\frac{\sigma_+ }{\zeta p - \sigma_+} ) \nu < 1 . $
  
  \item[(ii)] Suppose that there exists $\mu>0$ and $s>0$ such that
$
W_{\rm ns}(\bx) = \mathcal{O} (\exp(-\mu |\bx|^{s}),$ as $ |\bx|\rightarrow \infty\,.
$
Then,  for any $\nu>0$, there exists $C>0$ such that, for all $\vp\in E_{I}(H) L^2(\R^2)$, we have
\begin{equation}
 \SP{\vp(t)} { |\bx|^{\beta} \vp(t)} \le C \left(
  \SP{\vp}{|J|^{ \zeta \nu /\sigma_-} \vp}  +  \ln(t)^{\theta_s  }   \n{\vp}^2     \right)\,, \qquad t>1\,,
\end{equation}
where $\theta_s  = \frac{1 }{\sigma_- \min \{ 1 , s/\sigma_+  \}} \, \nu    $.
\end{enumerate}
\end{theorem}

\section{From tunnelling estimates to dynamical bounds}\label{section thm 2}
{  We start with the proof of Theorem \ref{thm 2} and consider  Theorem \ref{thm 1} at the end of this section. The proof of Theorem \ref{thm 2} is based upon certain dynamical bounds, that use Heisenberg's equation for $J(t)$, combined with the tunnelling estimates given in  Remark \ref{tunn estimate}. Before proceeding with the proof of Theorem \ref{thm 2} we establish the following. }
\begin{lemma}\label{Lns}
Assume that conditions \ref{con3} and \ref{con1} are
satisfied. Then, there is a constant $C_{\rm ns}\in (0,\infty)$ such
that
\begin{align*}
  \norm{E_I(H) W_{\rm ns}}\le C_{\rm ns}\,.
\end{align*}
\end{lemma}
\begin{proof}
   In view of Equation \eqref{asbe}, for $j=0$, and Remark~\ref{rel-bdd} we
  see that the inequality
  \eqref{eq:3} holds for $W_{\rm ns}$ as well as for $W_{\rm s}$. This
  implies that
  $\calD(W_{\rm ns})\supseteq \calD(H_0)=
  \calD(H)$.
  By the Closed Graph Theorem, we have that
  $ W_{\rm ns} (H+\lambda)^{-1}$ is a bounded operator, for some
  {  $\lambda >-e_0$}.  We conclude by observing that
  \begin{align*}
     \norm{ W_{\rm ns} E_I(H)}\le \norm{W_{\rm ns} (H+\lambda)^{-1}}
    \,\norm{ (H+\lambda) E_I(H)}< \infty\,. 
  \end{align*}
\end{proof}
\begin{proof}[Proof of Theorem \ref{thm 2}]
For any $\vp\in  E_{I}(H)L^2(\R^2)$ and  $M>0$ we have
\begin{equation}\label{estimate 1}
\n{ |J|^{\beta/2}  \vp(t)   }^2 \leq M^\beta \n{\vp}^2 + \sum\nolimits_{|j|>M} |j|^\beta \n{P_j \vp (t)}^2.
\end{equation}
Using Heisenberg's evolution equation we get
$$ \n{P_j \vp(t)}^2 = \n{P_j \vp}^2 + i \int_0^t \< \vp(s), [W,P_j] \, \vp(s)   \>  \d s\,.$$
Notice that $[W,P_j]=[W_{\rm ns},P_j]$. The above equation combined
with \eqref{estimate 1} yields
\begin{equation}\label{Jbeta}
\begin{split} 
\n{ |J|^{\beta/2}  \vp(t)   }^2 &\le \n{ |J|^{\beta/2}  \vp }^2 +
M^\beta \n{\vp}^2   + \sum_{|j|>M} |j|^\beta \int_0^t  |\<\vp(s),
[W_{\rm ns},P_j] \, \vp(s)   \>| \, \d s\\
&\le  \n{ |J|^{\beta/2}  \vp }^2 +
M^\beta \n{\vp}^2   + 2  \sum_{|j|>M} |j|^\beta \int_0^t |\<\vp(s),
W_{\rm ns}P_j \vp(s)   \>| \, \d s\\
&\le  \n{ |J|^{\beta/2}  \vp }^2 +
M^\beta \n{\vp}^2   + 2t  \sum_{|j|>M} |j|^\beta 
\norm{E_I(H)W_{\rm ns}P_j E_I(H)} \norm{\vp}^2 \,.
\end{split}
\end{equation}
Using Lemma~\ref{Lns} we can estimate the norm appearing in the last
sum  as 
\begin{equation}
  \label{eq:7}
\begin{split}
& \norm{E_I(H)W_{\rm ns}P_j E_I(H)} \\
   &\quad \le  \norm{E_I(H)W_{\rm ns}
   \mathds{1}_{ (0,c_+ |j|^{ \zeta/ \sigma_+} ) }(\ab{\bx}) P_j E_I(H)}
  + \norm{E_I(H)W_{\rm ns}  \mathds{1}_{(c_+ |j|^{ \zeta/ \sigma_+}, \infty)}(\ab{\bx})
    P_j E_I(H)} \\
  &\quad \le  \norm{E_I(H)W_{\rm ns}}\,
     \norm{ \mathds{1}_{ (0,c_+ |j|^{ \zeta/ \sigma_+} )  }(\ab{\bx}) 
       P_j E_I(H)}
        + \norm{W_{\rm ns}  
             \mathds{1}_{ (c_+|j|^{ \zeta/ \sigma_+}, \infty)}(\ab{\bx})}\\
&\quad \le  C  \norm{E_I(H)W_{\rm ns}} e^{-\delta \ab{j}^{\zeta}}+
 \norm{W_{\rm ns}  \mathds{1}_{ (c_+ |j|^{\zeta/ \sigma_+}, \infty)} (\ab{\bx})}\,,
\end{split} 
\end{equation}
where in the last step we used Remark~\ref{tunn estimate}.

We now study separately the two cases depending on the decay rate of
the potential.\\
{\it Case {\rm(}i\,{\rm)}}: Assume that $W_{\rm ns}$ decays as in \eqref{powerlike1}.
Then for all sufficiently large  $M>0$ 
\begin{equation}
  \label{eq:9}
  \norm{W_{\rm ns}  \mathds{1}_{ (c_+ |j|^{ \zeta/ \sigma_+}, \infty)}
    (\ab{\bx})}
      \lesssim |j|^{-{\zeta p/ \sigma_+}}\,, \quad \ab{j}>M\,. 
\end{equation}
Combining this bound with \eqref{eq:7}, using 
that $e^{-\delta \ab{j}^\zeta}\le |j|^{-{ \zeta p/ \sigma_+}}$ for large 
$\ab{j}$ this implies 
\begin{equation} \label{eq:10}
  \sum_{|j|>M} |j|^\beta \norm{E_I(H)W_{\rm ns}P_j E_I(H)}
    \lesssim   
      \sum_{|j|>M} |j|^{ \beta-{\zeta p/ \sigma_+} } 
    \lesssim  
      M^{{ \beta+1-\zeta p/\sigma_+ } }\,. 
\end{equation}
Since $\beta-\zeta p/\sigma_+<-1$. 
This together with \eqref{Jbeta} implies that there exists a constant $C>0$ such that
\begin{equation*}
 \n{ |J|^{\beta/2}  \vp(t)   }^2 \leq \n{ |J|^{\beta/2}  \vp   }^2 +
M^\beta \n{\vp}^2 + C t M^{{ \beta+1-\zeta p/\sigma_+ } } \n{\vp}^2\,. 
\end{equation*}
The theorem holds provided we pick $M(t) = t^{{ \frac{\sigma_+}{\zeta p-\sigma_+}}} $.
\\
{\it Case {\rm(}ii\,{\rm)}}: We now turn  to the case when $W_{\rm ns}$ decays as
in \eqref{gaussian}. Then analogously as above we conclude for all $\ab{j}$ large enough 
\begin{align*}
  \norm{E_I(H)W_{\rm ns}P_j E_I(H)} 
    &\lesssim  
        \exp(- \delta|j|^\zeta) +  \exp ( - \mu (c_+)^s |j|^{s { \zeta/\sigma_+} })
       \lesssim  
         \exp(- \eta |j|^{\kappa})\,,
\end{align*}
for $\kappa:={\rm min}\{\zeta, { s\zeta/\sigma_+} \}$ and $\eta= \min\{\delta,  \mu (c_+)^s \}>0$. Hence 
\begin{equation}
  \label{eq:10-2}
  \sum_{|j|>M} |j|^\beta 
\norm{E_I(H)W_{\rm ns}P_j E_I(H)}\lesssim  \exp(- \tfrac{\eta}{2} M^{\kappa})\,. 
\end{equation}
for all large enough $M>0$. 
This together with \eqref{Jbeta} yields
\begin{align*}
  \n{ |J|^{\beta/2}  \vp(t)   }^2 \leq \n{ |J|^{\beta/2}  \vp   }^2 +
  M^\beta \n{\vp}^2  + 2 C t \exp(- \tfrac{\eta}{2} M^{\kappa})\n{\vp}^2\,. 
\end{align*}
for some constant $C>0$. 
Finally, the claim follows by picking $M(t) = (2 \ln(t)/\eta)^{1/\kappa}$.
\end{proof}
\begin{remark}
 { The proof of Theorem \ref{thm 1} is a rigorous implementation of  
  the intuition that the component of the wave--function with angular 
  momentum $j$, $P_j\varphi$, moves under the influence of an effective potential $V_j$ whose classical region $\{x\in \R^2 \, :\,  V_j(x) \leq E \}$ is concentrated inside an annular region of inner radius  
  $\sim |j|^{{ 1 /\sigma_+}}$ and outer radius    $\sim |j|^{{ 1 /\sigma_-}}$ 
  (see Theorem \ref{thm-tunnelling} and Remark \ref{tunn estimate2}).}
\end{remark}

\begin{proof}[Proof of Theorem \ref{thm 1}:]
	Let $t\ge 0$ and consider the following splitting
\begin{equation*}
\n{ |\bx|^{\nu/2} \vp(t)   }^2 
  = \sum_{j\in\Z} 
      \n{ \mathds{1}_{[0,c_- |j|^{\zeta /\sigma_-} ]} \, |\bx|^{\nu/2} P_j E_I(H) \vp(t)   }^2 
    + \sum_{j\in\Z} \n{ \mathds{1}_{[c_- |j|^{\zeta /\sigma_-} ,\infty)} \, |\bx|^{\nu/2} P_j E_I(H) \vp(t)    }^2
\end{equation*}
where $c_-$ is the constant given by {  Theorem} \ref{thm-tunnelling}. 
For notational simplicity, we will drop the argument in the 
indicator function and simply  write
$\mathds{1}_{A} \equiv \mathds{1}_{A}(\ab{\bx})$ in the following. 
The first sum may be estimated in terms of the expectation values of the angular momentum as
$$ \sum_{j\in\Z} \n{ \mathds{1}_{[0,c_- |j|^{\zeta/\sigma_-} ]} \, |\bx|^{\nu/2}
  P_j E_I(H) \vp(t) }^2 \leq c_-^\nu \n{ |J|^{{ \nu\zeta/(2 \sigma_-)}}
  \vp(t)}^2.  $$
As for the second sum, we may use Equation \eqref{lemma 1.1 b} from
{  Theorem} \ref{thm-tunnelling} to obtain
\begin{align*}
\sum_{j\in\Z} & 
  \n{ \mathds{1}_{[c_- |j|^{\zeta/\sigma_-} \!, \infty)} \,
  |\bx|^{\nu/2} P_j E_I(H) \, \vp(t)    }^2  \\  
 &\le    
   \n{|\bx|^{\nu/2} e^{-\delta |\textbf{x}|^{ \zeta\sigma_-}}  }^2
     \sum_{j\in\Z} \| 
       \mathds{1}_{ [c_- |j|^{{\zeta/\sigma_-}} ,\infty )  } \, e^{\delta|\t{\textbf{x}}|^{\zeta\sigma_-} } \, P_j E_I(H)\|^2 \n{\vp}^2.
\end{align*}
This completes the proof of  the theorem.
\end{proof}

%%%%%%%%%%%%%%%%%%%%%%%%%%%%%%%%%
%%%%%%%%%%%%%%%%%%%%%%%%%%%%%%%%%
%%%%%%%%%%%%%%%%%%%%%%%%%%%%%%%%%
%%%%%%%%%%%%%%%%%%%%%%%%%%%%%%%%%
\section{Exponential decay of the energy projections}\label{edcay}
{  An essential tool in our approach are certain exponential decay
estimates for the spectral projections $E_I(H)$ in the variables $j$ and $r$. They enable  us  to control the tunnelling effect away from
the classically allowed region. Note that since the perturbation $W$ is not rotationally symmetric, $H$ and $J$ cannot be simultaneously diagonalized. This combined with the fact that the unperturbed operator has dense point spectrum makes such decay bounds trickier to deal with.  

It turns out that it suffices to work with the classical region associated to the unperturbed magnetic operator $H_0$, which is well defined for 
fixed angular momentum $j$.
The analysis in this section is given for general magnetic fields which are rotationally symmetric such that the magnetic vector potential 
vector potential in the Poincar\'e gauge is locally square integrable.

For a given energy $E>0$ and fixed $j\in \Z$ we define the {\it classically allowed region} for angular
momentum $j\in \Z$ as the set 
\begin{align}
\label{class allowed}
\mathcal{C}_j(E):=\{r\in \R^+\,:\,V_j(r)\le E\}\,.
\end{align}
where $V=(V_j)_{j\in\Z}$ is the effective potential. 
Moreover, let $\chi_j(E): \R^+\to [0,1]$, with $\chi_j(E)=1$ on
$\mathcal{C}_j(E)$ and $\chi_j(E)=0$ otherwise, be the indicator
function on $\mathcal{C}_j(E)$. We also  set $\chi_j^\perp (E):=1-\chi_j(E)$.

Recall the constants
 $a>0$ and $\zeta\in (0,1]$ are defined through Condition~\ref{con3}
 and write $ \xi(a,\zeta)=\sum_{m\in\Z} e^{-\frac{a}{2}|m|^\zeta}$. 
 Since $v$ is a radial potential, Lemma \ref{lem radial energy boundedness} in Appendix \ref{app-polar coordinates} shows 
 that for any $a>0$ and $0<\zeta\le 1$ there exist a $c_0$, such that 
 \begin{align}
  \label{eq:13}
  \SP{\partial_r\vp}{\partial_r\vp}- \xi(a,\zeta) \SP{\vp}{v\vp}\ge -c_0\n{\vp}^2\,,
\end{align}
for all $\vp$ in the quadratic form domain of the magnetic Schr\"odinger operator.

We denote by $PC^1_\text{bd} (\R^+,\R )$ the set of \emph{bounded} functions $f: \R^+
\rightarrow \R$ that are  continuous and piecewise continuously differentiable.
We say that a sequence of non-negative functions
	$F=(F_j)_{j\in\Z}$ is in $PC^1_\text{bd} (\R^+,\R )$ whenever
	$F_j\in PC^1_\text{bd} (\R^+,\R )$ and 
	\begin{align}
		\|F\|_\infty\coloneqq \sup_{j\in\Z}\sup_{r>0} |F_j(r)| <\infty\, .
	\end{align}

	For such sequences we write $e^{\pm F}=\sum\nolimits_{j\in\Z} e^{\pm F_j} P_j$, which are bounded operators satisfying $\| e^{ \pm F} \| \leq e^{\|F\|_{\infty }} $. Moreover, Lemma \ref{lem-twisted q_0} shows that if also $F'$ is bounded, i.e., 
	\begin{align}
				\|F'\|_\infty\coloneqq \sup_{j\in\Z}\sup_{r>0} |F_j'(r)| <\infty
	\end{align}
	then $e^{\pm F}\vp$ is in $\calD(q_0)=\calD(q)$, the domain of the magnetic Schr\"odinger, whenever $\vp\in \calD(q_0)$.
	We also set 
	$\chi(\tilde{E})  
	 \coloneqq \sum\nolimits_{j\in\Z} \chi_j(\tilde{E}) P_j$.
	
\begin{theorem}[Exponential decay of energy projections]\label{thm 3}
 Let $H$ be the perturbed magnetic Schr\"odinger operator defined by the quadratic form \eqref{def-magn form}. 
 Given $\delta_0>0$ and $E_0\ge 0$ put 
\begin{equation}
\label{E}
\tilde{E}:=E_0+c_{0}+\delta_0\,.
\end{equation}
with $c_0$ from \eqref{eq:13} above.  
 Then, for any sequence of weight functions  $F=(F_j)_{j\in\Z}$ in $PC^1_{\text{bd}} (\R^+,\R )$ that satisfy 
 \begin{align}
 	\label{F'} &(F')^2 \le V-\wti{E}\chi^\perp(\wti{E}) \, ,\\
 	\label{Fbounded}&\|e^F\chi(\wti{E})\|_\infty <\infty\, , \\
 		 \label{efw} 
	& \sup_{r>0}\ab{F_j(r) - F_k(r)}\leq \tfrac{a}{2}\ab{j-k}^\zeta,   \t{
          for all  }   j,k\in\Z\,, 
 \end{align}
there exists $C = C (\delta_0)>0$ such that 
 \begin{equation}\label{eqexp}
	\norm{e^F \,E_I(H)} \le C \|e^F\chi ( {\wti{E}} ) \|_\infty \, .
	\end{equation} 
%  Let $H=H_0+W$, where $H_0$ is defined through the quadratic form given in \eqref{eq:22}, for $\mathbf{A}\in L_{\rm loc}^2(\R^2, \R^2)$, and $W$ satisfying Condition \ref{con3}.
% Let $\delta_0> 0$ be fixed. 
%Let $\epsilon_0\in [0,1)$ and 
%$c_{0}\ge 0$ be fixed constants satisfying \eqref{eq:13}. Define 
%\begin{equation}
%\label{E}
%\tilde{E}:=(E_0+c_{0}+\delta_0)/(1-\epsilon_0)\,.
%\end{equation}
%	Then there exist a constant $C_{\delta_0}<\infty$ such that  for any  sequence of non-negative functions
%        $(F_j)_{j\in\Z}\subseteq PC^1 (\R^+,\R )$ satisfying the
%        following conditions
%	\begin{align}
%	 \label{F'}
%	& (F_j')^2\leq (1-\epsilon_0)(V_j - \tilde{E} )\chi_{j}^{\perp}(\tilde{E}),   \t{ a.e on } \R^+ \t{ for all } j\in\Z\,, \\
%	 \label{Fbounded} 
%	& F_j \,\chi_j (\tilde{E}) = 0, \t{ for all but finitely many } j\in\Z\,, \\
%	 \label{efw} 
%	& \sup_{r>0}\ab{F_j(r) - F_k(r)}\leq \tfrac{a}{2}\ab{j-k}^\zeta,   \t{
%          for all  }   j,k\in\Z\,, 
%	\end{align}
%	we have 
%	\begin{equation}\label{eqexp}
%	\norm{e^F \,E_I(H)}^2=\sum_{j\in \Z} \n{ e^{F_j(r)} P_j  E_I(H)}^2  <\infty\,.
%	\end{equation}
 If $W=0$ the bound \eqref{eqexp}  holds without the requirement of \eqref{efw}. 
 
 Moreover, the bound \eqref{eqexp} extends to all  $F=(F_j)_{j\in\Z}$ in $ PC^1(\R^+,\R )$ satisfying conditions \eqref{F'}, \eqref{Fbounded}, and \eqref{efw} without requiring that  $F$ is bounded.  
	\end{theorem} 
%%%%%%
%%%%%%
%%%%%% 
\begin{remarks}
\begin{remarklist}
	\item Note that the right hand side of \eqref{eqexp} stays finite as long as 
	$e^F$ is bounded on $\supp(\chi(\wti{E})) =\cup _{j\in \Z } \, \supp (\chi_j (\wti{E} ) )  $.  Thus, we may approximate 
	unbounded weight functions $F$, which may grow in  both variables 
	$r>0$ and $j\in \Z$, by bounded ones and deduce from \eqref{eqexp} 
	that $\norm{e^F \,E_I(H)} $ is \emph{finite} as long as $F$ is bounded on the support of $\chi(\wti{E})$.  
		In particular, we will choose a family $(F_j)$ such that it
	provides exponential decay estimates of $E_I(H)$ away from the classical regions
	$\mathcal{C}_j(\wti{E})$. The difference  $\wti{E}-E_0>0$,  is a price we pay for defining the classical region with respect  to the operator  $H_0$
	instead of $H$.
  \item \label{delta-rmk}
  In the case when $W=0$ we may choose $c_0=0$ and
  we have almost optimality in the energy
  $\wti{E}=E_0+\delta_0$. 
 \item  
	Our proof of  Theorem \ref{thm 3} borrows ideas from
\cite{Griesemer2004} which are, in turn, inspired by the proof of
exponential decay in QED systems given in  \cite{BFS1998b} and 
the beautiful approach to exponential bounds for eigenfunctions in 
\cite{agmon}.   
\end{remarklist}
\end{remarks}
As a first step we define a more convenient (smooth) version of the spectral projection $E_I$, denoted by $g_{\triangle}(H)$, which is given in formula \eqref{zero} below. The constants $\wti{E}, E_0, \delta_0$ and $c_0$ are as in Theorem $\ref{thm 3}$.

Set
$\triangle:=[e_0-\delta_0/2,E_0+\delta_0 /2]$, where $e_0=\inf\sigma(H)$.
Consider $g_{\triangle}\in
C_0^{\infty}(\R,[0,1])$ such that $\supp \, g_{\triangle}\subseteq
\triangle$ and $g_{\triangle}|_{I} =1 $.
Since $E_I(H)= g_{\triangle}(H)E_I(H)$ it is enough to prove the bound 
\eqref{eqexp} with  $E_I(H)$  replaced by $g_{\triangle}(H)$. 

Next, we use  the almost analytic functional calculus (see
\cite{Davies1995}) to write $g_{\triangle}(H)$ in terms of an integral
over the resolvent of $H$. We denote by $\tilde{g}_{\triangle}$
an almost analytic extension of $g_{\triangle}$ with the property that 
${\rm supp}(\tilde{g}_{\triangle})$ is a compact subset of 
$\triangle+i\R $ and 
\begin{align}
\label{eq:2}
\left| \partial_{\bar{z}}\tilde{g}_\triangle (z)\right|
=\mathcal{O}(|{\rm Im}(z)|), \quad {\rm Im}(z)\to 0\,.
\end{align}
One can give a straightforward explicit construction of  
$\tilde{g}_\triangle$ with the above properties, however, 
see \cite{Davies1995,Griesemer2004} for details.  
Then, we have the formula 
\begin{equation}
\label{almost}
g_{\triangle}(H)=-\frac{1}{\pi}\int (z-H)^{-1} \,\partial_{\bar{z}}\tilde{g}_\triangle(z)  \,\d x \, \d y\,.
\end{equation}
for any self-adjoint operator $H$.  
We work with the comparison operator, more precisely, with the quadratic form corresponding to 
\begin{align*}
&\wti{H} \coloneqq H + \wti{E} \chi(\wti{E})\,,
\end{align*}
Notice that $\wti{H}$ is just $H$ except that 
it is boosted by $\wti{E}$ in (a neighborhood of) the classical region, i.e., as quadratic forms 
\begin{align}\label{eq-superdooper}
	\SP{\vp}{\wti{H}\vp} 
	  = \SP{\partial_r\vp}{\partial_r\vp} 
	     + \SP{\vp}{(V+ \wti{E} \chi(\wti{E}))\vp} + \SP{\vp}{W\vp}
\end{align}
for all $\vp\in \calD(q_0)$
 
It is easy to check, see Lemma \ref{lemm 3} below, 
that $\wti{H}\ge  E_0+\delta_0$. Since $\supp g_{\triangle} \subset \triangle$, we have that
$g_{\triangle}(\wti{H})=0$ and, therefore, by the resolvent
identity  and \eqref{almost} 
\begin{equation}
\label{zero}
g_{\triangle}(H) = g_{\triangle} (H)-g_{\triangle}(\wti{H}) =
-\frac{1}{\pi}\int  (z-\wti{H})^{-1} (\wti{H} - H) (z-H)^{-1}    \,\frac{\partial\tilde{g}_\triangle}{\partial \bar{z}}     \,\d x \, \d y\,.
\end{equation} 
%where $\wti{H} - H= \wti{E}\chi(\wti{E})$ is a bounded operator with `support' near the `classically allowed region'.  

\begin{lemma}\label{lemm 3}
	For the operator $\wti{H}$ defined above one has a quadratic forms 
	\begin{align}
	\label{lemm 3.1}
	\wti{H} \ge  E_0+\delta_0\,.
	\end{align}
	Furthermore, if the sequence of non-negative functions
	$F=(F_j)_{j\in\Z}$ in  $PC^1_\text{bd} (\R^+,\R )$ satisfies   \eqref{F'}, \eqref{efw}, and $F'$ is bounded, then   
	\begin{align}
	\label{lemm 3.2}
	\sup_{z\in\supp
		(\tilde{g}_{\triangle}) }  
	\| e^{F}(z-\wti{H})^{-1}e^{-F} \| \le 2/\delta_0\,. 
	\end{align}
	where $\delta_0>0$ is the parameter from Theorem \ref{thm 3}.
\end{lemma}
\begin{remark}
  Again, it is essential, that the right hand side of \eqref{lemm 3.2} does not depend on any a-priori bound on $\|F\|_\infty$. 
\end{remark}
\begin{proof}

%	For shorthand notation we write $\chi_j(\tilde{E})\equiv \chi_j$ and
%	$\chi(\tilde{E})\equiv \chi$.  
	
	In order to show \eqref{lemm 3.1}
	notice that from \eqref{eq-superdooper} we get 
  \begin{equation}
  \begin{split}	\label{eq:15}
	\SP{\vp}{\wti{H}\vp} 
	  &= \SP{\partial_r\vp}{\partial_r\vp} 
	     + \SP{\vp}{(V+ \tilde{E} \chi(\tilde{E}))\vp} + \SP{\vp}{W\vp} \\
	  &\ge \SP{\partial_r\vp}{\partial_r\vp}  + \SP{\vp}{(V+ \tilde{E} \chi(\tilde{E}))\vp} - \xi(2a,\zeta)\SP{\vp}{v\vp} \\
	  &\ge    \SP{\vp}{(V+ \tilde{E} \chi(\tilde{E})-c_0)\vp}  
	       =     \SP{\vp}{(V- \tilde{E} \chi^\perp(\tilde{E}) +\wti{E}-c_0)\vp}  
  \end{split}
  \end{equation}
  	where we also used the bound \eqref{relbound} from
	Lemma~\ref{w-and-ef}  and then  \eqref{eq:13}. Since by assumption 
	$V\ge   \wti{E}\chi_j^{\perp}$ and $\wti{E}-c_0=\delta_0$, the bound 
	\eqref{lemm 3.1} follows. 
	
	Next we turn to the proof of \eqref{lemm 3.2}. 
 Since $\|F\|_\infty, \|F'\|_\infty<\infty$, we know that $e^{\pm F}$ 
 are bounded operators $L^2(\R^2)$ which by Lemma \ref{lem-twisted q_0} also leave the form domain of $H$ invariant.

Moreover, the operator $\tilde H_F:=e^F\wti{H} e^{-F}$, or better its 
associated quadratic form, is well defined, see the discussion in 
Appendix \ref{app-twisted}. It is easy to check that $ \wti{H}_F-z$ 
is invertible if $z$ is such that 
	\begin{equation}
	\label{eq:16}
	\textrm{Re} \,  \sps{\varphi}{( \wti{H}_F-z)\varphi}\ge \tfrac{\delta_0}{2} \norm{\varphi}^2\,, \qquad \varphi \in \calD( \wti{H}_F)\,,
	\end{equation}
	and that then also 
	\begin{align*}
		\n{(\wti{H}_F-z)^{-1}}\le \frac{2}{\delta_0}
	\end{align*}
  since clearly $ \textrm{Re} \,  \sps{\varphi}{( \wti{H}_F-z)\varphi}
	\le\norm{( \wti{H}_F-z)\varphi} \norm{\varphi} $.
  If \eqref{eq:16} holds, then it is also straightforward to show 
  	\begin{align*} 
	e^F (\wti{H}-z)^{-1} e^{-F}=( \wti{H}_F-z)^{-1} \,.
	\end{align*}
Thus, 	to show \eqref{lemm 3.2} it suffices to prove that, for any $z\in {\rm
		supp}\tilde{g}_{\triangle} $, \eqref{eq:16} holds. For this 
		we use the exponentially twisted version of 
		\eqref{eq-superdooper} provided by \eqref{eq-superdooper2} 
		in Appendix \ref{app-twisted} which shows that as 
		quadratic forms  
	\begin{align*}
	\re \sps{\varphi}{e^F H_0 e^{-F}\varphi}
	  &= \SP{\partial_r\vp}{\partial_r\vp} + \SP{\vp}{(V+\wti{E}\chi(\wti{E})-(F')^2)\vp} 
		      + \re \SP{e^F\vp}{We^{-F}\vp}  \\
	  &\ge \SP{\partial_r\vp}{\partial_r\vp} + \SP{\vp}{(V+\wti{E}\chi(\wti{E})-(F')^2))\vp} 
		      -   \xi(a,\zeta) \SP{\vp}{v\vp} \\
	  &\ge  \SP{\vp}{(V+\wti{E}\chi(\wti{E})-(F')^2-c_0)\vp}  \\
	  & =  \SP{\vp}{(V-\wti{E}\chi^\perp(\wti{E})+\wti{E}-(F')^2-c_0)\vp}
	\end{align*}
	where we also used 	Lemma~\ref{w-and-ef}  and  \eqref{eq:13}.
	Thus, using also \eqref{F'} we get 
  \begin{align*}
	\re \sps{\varphi}{e^F H_0 e^{-F}\varphi}
	  & \ge   \SP{\vp}{(V-\wti{E}\chi^\perp(\wti{E})+\wti{E}-(F')^2-c_0)\vp} \\
	  &= \SP{\vp}{(V-\wti{E}\chi^\perp(\wti{E})-(F')^2+E_0+\delta_0)\vp}\\ 
	  &\ge (E_0+\delta_0)\|\vp\|^2
	\end{align*}
 which implies \eqref{eq:16} since for all $z\in \supp(\wti{g}_\Delta)$ we have 
 $\textrm{Re} \,z\le E_0+\delta_0/2$. 
	This finishes the proof of the lemma.
\end{proof}
\begin{proof}[Proof of Theorem \ref{thm 3}]
Using equation \eqref{zero} we may write 
\begin{align}\label{eg}
e^{F} g_{\triangle}(H)=-\frac{\wti{E}}{\pi}\int  e^{F} (z-\wti{H})^{-1}e^{-F}\, e^{F}\chi{\wti{E}}  (z-H)^{-1}    \,\frac{\partial\tilde{g}_\triangle}{\partial \bar{z}}     \,\d x \, \d y\,,
\end{align}
so 
\begin{align}\label{eg2}
  	\n{e^{F} g_{\triangle}(H)}
  	  &\le \frac{\wti{E}}{\pi}\n{e^{F}\chi(\wti{E})} \int  \n{e^{F} (z-\wti{H})^{-1}e^{-F}}\, \n{  (z-H)^{-1}}    \,
  	  		|\partial_{\bar{z}}\tilde{g}_\triangle(z)| \,\d x \, \d y\,, 
  	  		\\
  	  &\le \n{e^{F}\chi(\wti{E})}   \frac{2\wti{E}}{\pi\delta_0} \int \frac{|\partial_{\bar{z}}\tilde{g}_\triangle(z)|}{|y|}\, dx dy
\end{align}
where we also used the  standard bound $\norm{ (z-H)^{-1} }\le 1/|{\rm Im}(z)|$. Since for the almost analytic extension the integral above is finite,  this  proves  \eqref{eqexp} when $F$ and $F'$ are bounded. 

Now assume that $F$ is positive but unbounded and satisfies \eqref{F'}, \eqref{Fbounded}, \eqref{efw}, and $F'$ bounded.  
Define  the family  of bounded functions $F_n=(F_{j,n})_{j\in\Z}$ given by 
\begin{align}
\label{eq:18}
F_{j,n}:=\frac{F_j}{1+\tfrac{1}{n} F_j}\,,\quad n \in \N \,, j\in \Z\,.
\end{align}
Notice that $|F_{j,n}'|= (1+\tfrac{1}{n}F_j)^{-2}|F_{j,n}'|$ is bounded uniformly in $j\in\Z$.  
Clearly, for each $j\in \Z$, the sequence $(F_{j,n})$ is
increasing in $n\in\N$ and converges to $F_j$ 
pointwise. Moreover,
one checks that  for each $n\in \N$ we also have $\ab{F_{j,n}-F_{k,n}}\le
\ab{F_j-F_k}$, hence $(F_{j,n})_{j\in \Z}$ satisfies conditions  \eqref{F'},
\eqref{Fbounded} and \eqref{efw} for each $n\in\N$. 
Moreover, notice that for any $n\in \N$, the estimate $F_{j,n}\le n$  holds uniformly in $j\in \Z$. 

By what we just showed, this implies that the bound \eqref{eqexp} holds when $F$ is replaced by $F_n$, but since the right hand side of 
\eqref{eqexp} is uniform in $n\in\N$, the Monotone Convergence Theorem shows that  for any $\vp\in L^2(\R^2)$, 
\begin{align*}
\norm{e^F g_{\triangle}(H) \varphi}^2&=\lim_{n\to \infty} \norm{e^{F^{(n)}}g_{\triangle}(H) \varphi}^2 
\le C_{\delta_0}^2\norm{\varphi}^2\,.
\end{align*}
This finishes the proof. 
\end{proof}
}
\section{The tunnelling bounds}\label{aplications}
In this section we apply Theorem \ref{thm 3} to derive the interior and exterior tunnelling bounds from Theorem  \ref{thm-tunnelling}.
To do so, we need to construct  suitable sequences of weights 
$(F_j)_{j\in\Z}$ that satisfy the requirements of Theorem \ref{thm 3}.

In order to verify  \eqref{F'} it is important to estimate the value of the effective potential $V_j(r)=( j-\Phi(r))^2/r^2$ in the classically forbidden  regions. One can get an intuition by considering the case $\Phi(r)=r^{\sigma}$, $\sigma>1$. It is easy to see that  the classically allowed region is either empty, or it is contained in an interval $[r_-|j|^{1/\sigma} , r_+|j|^{1/\sigma}]$, for some $r_+>r_->0$.  Conditions~\ref{con1} and \ref{con1b} allow us to obtain estimates for $V_j-E$, to the left and to the right of the classically allowed region of the corresponding effective potential, respectively. This is shown in the next lemma, for any energy $E>0$.
{
\begin{lemma}\label{cfr}
	\begin{enumerate}[i{\rm)}]
		\item\label{ec1}  Under Condition 
			\ref{con1} there exists a constant 	$j_0>0$ and, given $E>0$, 
			a constant  $\varepsilon_E>0$ such that,  for any $j\in \Z$ with 
			$|j|\ge j_0$ and $r\le  \varepsilon_E |j|^{1/\sigma_+}$ 
		  \begin{align}\label{bound1}
		  V_j(r) -E\ge  |j|^{2\frac{\sigma_+-1}{\sigma_+}}\,.
		  \end{align}
		\item\label{ec2}  Under Condition \ref{con1b} 
		there exist for any $E>0$, a constant 
		$\eta_E>1$ such that, for any $j\in \Z$ and $r
		\ge \eta_E (1+|j|)^{1/\sigma_-}$,  
		 \begin{align}\label{bound2}
		V_j(r) -E\ge  \lambda_-^2 r^{2(\sigma_--1)}\,,
		\end{align}
		where $\sigma_\mp, \lambda_\mp$ are parameters defined in  Condition \ref{con1} and \ref{con1b}.
	\end{enumerate}
\end{lemma}
} 
\begin{proof}
\ref{ec1}): From Condition \ref{con1} one sees that for any  
  $\varepsilon>0  $ we have  for all $r\le \varepsilon |j|^{1/\sigma_+}$
\begin{align}
|\Phi(r)|\le \lambda_+(1+r^{\sigma_+}) \le  \lambda_+(1+ \varepsilon^{\sigma_+} |j|\big)  \,. 
\end{align}
Thus 
\begin{align*}
  \sqrt{V_j(r)}
    &=\frac{1}{r} \left|j-\Phi(r) \right|
       \ge \frac{1}{r} \left(  |j|-|\Phi(r)| \right) 
		\ge \frac{1}{r}\big(|j|- \lambda_+(1+\varepsilon^{\sigma_+}|j|\big)\big) \\
	&= \frac{1}{r}\big( |j|(1-\lambda_+\varepsilon^{\sigma_+})-\lambda_+ \big)
\end{align*}
for all $0<r\le \varepsilon |j|^{1/\sigma_+}$. 
Thus, if we choose $\varepsilon $ so small that $1-\lambda_+\varepsilon{^\sigma_+}\ge 1/2$ we get 
\begin{align}
  \sqrt{V_j(r)}
   	&\ge  \frac{1}{r}\big( \frac{1}{2}|j|-\lambda_+ \big) 
   	      \ge \frac{|j|}{4r} \ge \frac{1}{4\varepsilon} |j|^{1-\frac{1}{\sigma_+}}
\end{align}
whenever $0\le r\le \varepsilon |j|^{1/\sigma_+}$ and $|j|\ge 4\lambda_+$. So by making $\veps$ so small that also 
$16\varepsilon^2\le  1/(E+1)$ we get, for any $|j|\ge 4\lambda_+>0$ and 
$r\le \varepsilon |j|^{1/\sigma_+}$, 
\begin{align}
{V_j(r)}-E\ge |j|^{2\frac{\sigma_+-1}{\sigma_+}}\left( \frac{1}{16\varepsilon^2}-E\right)\ge  |j|^{2\frac{\sigma_+-1}{\sigma_+}}\,.
\end{align}
This shows the claim.
\medskip

\noindent
\ref{ec2}):
%\begin{align*}
%	|\Phi(r)|\ge \lambda_- r^{\sigma_-}
%\end{align*}
%for all $r\ge r_0$. 
Let $\eta\ge \max\{r_0, (2/\lambda_-)^{1/\sigma_-}\}$ and $j\in \Z$. 
 Using Condition  \ref{con1b}, we see that  for any $r\ge \eta(1+|j|)^{1/\sigma_-}$
\begin{align}
  |\Phi(r)|\ge  \lambda_- \eta^{\sigma_-}(1+ |j|)\ge 2|j|\,.
\end{align}
Therefore, 
 \begin{align}
 \sqrt{V_j(r)}=\frac{1}{r} \left|\Phi(r)-j \right|\ge \frac{1}{r} \left( |\Phi(r)|-|j|\right)\ge 
 \frac{|\Phi(r)|}{2r}\ge \frac{\lambda_-}{2}  r^{\sigma_--1}\,.
 \end{align}
 Thus, for any  $r\ge \eta(1 +|j|)^{1/\sigma_-}$,
  \begin{align}
 V_j(r)-E\ge  r^{2(\sigma_--1)} (\lambda_-^2/4-E/\eta)\,, 
 \end{align}
 so claim follows with the choice $\lambda_0^2= \lambda_-^2/4-E/\eta$ by 
 choosing $\eta>1$ sufficiently large.
\end{proof}
\begin{remark}\label{obs}
The following simple observation is useful: If for any $j\in \Z$,  the sets $M_j$ are subsets of $ \R^+$, possibly empty, and $V_j-E>0$ on $M_j$, then 
  \begin{align}\label{clas}
&  \mathds{1}_{M_j}(r ) \chi_j(E)=0\,, \quad\mbox{and}\\\label{clas1}
   & (V_j(r)-E)\chi^\perp_j(E)\ge (V_j(r)-E)\mathds{1}_{M_j}(r )\,.
   \end{align} 
\end{remark}

Now we come to the 
\begin{proof}[Proof of Theorem~\ref{thm-tunnelling}]
	In order to show  \eqref{lemma 1.1 a} and \eqref{lemma 1.1 b} we construct  two different sequences and verify that they    satisfy the requirements of Theorem
	 \ref{thm 3}, equations   \eqref{F'}, \eqref{Fbounded} and \eqref{efw}. Throughout this proof we abbreviate  $\wti{E}\equiv E$ and $\chi_j(\wti{E})\equiv \chi_j$.
	
	\vspace{5mm}
	\noindent
	Proof of  \eqref{lemma 1.1 a}. Let $\varepsilon\in (0,\varepsilon_E]$ be 
	a constant to be fixed below. 
	In the following  $\varepsilon_E$ and $j_0>0$ are the parameters 
	from the first part of  Lemma \ref{cfr}. 
	We define, for any $j\in \Z$ with $|j|\ge j_0+1$,    
	\begin{equation}
	  F_j(r) =|j|^{\zeta(1-\sigma_+^{-1})} (\veps  |j|^{\zeta/\sigma_+} - r)_+\,, \qquad r>0\,,
	\end{equation}
    where  $x_+ = \max\{x,0  \}$, and $F_j(r)=0$ for all $r>0$ 
	when $|j|\le  j_0$. Note that $F_j$ is  piecewise continuously differentiable and its derivative is given by
	\begin{equation}\label{Fj'}
	|F_j ' (r)| =   |j|^{\zeta(1-\sigma_+^{-1})}\mathds{1}_{(0,\varepsilon |j|^{\zeta/\sigma_+})}(r)\,,
	\end{equation}
	when $|j|\ge j_0+1$ and its derivative vanishes when $|j|\le j_0$. 
	
	In view of \eqref{bound1} and \eqref{clas1} this choice of $F_j$ clearly 
	satisfies \eqref{F'},  since on $\R^+$ we have
		\begin{equation*}
(V_j-E)\chi_j^\perp\ge (V_j-E) \mathds{1}_{(0,\varepsilon |j|^{\zeta/\sigma_+})}\ge  |j|^{2\zeta( 1-\sigma_+^{-1})}\mathds{1}_{(0,\varepsilon |j|^{\zeta/\sigma_+})}\ge|F_j'|^2\,.
\end{equation*}
Moreover, using Remark \ref{obs},  one sees $F_j\chi_j=0$, for all $j\in \Z$, and hence we also have \eqref{Fbounded}.

 To show that $F$ satisfies  \eqref{efw}, it is enough to assume $|j|>|k|$, by symmetry. Also, if $|j|,|k|\le j_0$, then $F_j=F_k=0$, so \eqref{efw} trivially holds in this case.

 In the case    
 $|j| \geq |k| \ge j_0+1$ we  argue as follows: 
 If  $r\le \varepsilon |k|^{\zeta/{\sigma_+}}$, then we have  
 $$\(|j|^{\zeta(1-{\sigma_+}^{-1})} - |k|^{\zeta(1-{\sigma_+}^{-1})} \) r\le  \ve\: \big(|j|^\zeta-|k|^\zeta\big)\le \veps\: |j-k|^\zeta\, , $$ 
 since $0<\zeta\le 1$ and the map $\Z\ni j\to |j|^\zeta$ obeys the triangle inequality -- this is recalled in Lemma \ref{lem triangle inequality} in the appendix. Thus 
\begin{align*}
 |F_j(r) - F_k(r)| 
    &=  \veps\: \big(|j|^\zeta-|k|^\zeta\big) 
         - \(|j|^{\zeta(1-{\sigma_+}^{-1})} - |k|^{\zeta(1-{\sigma_+}^{-1}) } \) r \leq \veps |j-k|^\zeta\,.
		\end{align*}	
When  $r\in [\ve |k|^{\zeta/\sigma_+}, \ve |j|^{\zeta/\sigma_+}] $ then  
\begin{align*}
	 |F_j(r) - F_k(r)| 
	  &=  \ve \:|j|^\zeta - |j|^{\zeta(1-{\sigma_+}^{-1})}r 
	      \leq \ve |j|^\zeta - \ve |j|^{\zeta(1-{\sigma_+}^{-1})} |k|^{\zeta/\sigma_+} \\
	  &\le \ve (|j|^\zeta-|k|^\zeta) \le \veps |j-k|^\zeta  \,.
	\end{align*}
Moreover,  for $r> \ve |j|^{\zeta/\sigma_+}$ both $F_j$ and $F_k$ vanish, thus \eqref{efw} will hold for all $|j|, |k|\ge j_0+1$, provided we pick $\ve < a/2 $.

If $|j|\ge j_0+1$ and $|k|\le j_0$, then  $|j-k|\ge 1$, so 
\begin{equation*}
	|j| \le  |j-k|+|k|\le |j-k|+ j_0 \le |j-k|+ j_0|j-k| 
	   = (j_0+1)|j-k|
\end{equation*}
Thus, in this case 
\begin{align*}
	|F_j(r)- F_k(r)|= |F_j(r)| 
	  \le \veps |j|^\zeta
	  \le \veps (j_0+1)^\zeta|j-k|^\zeta \, .
\end{align*}
Choosing $\veps= a/(2(j_0+1)^\zeta)\le  a/2 $ 
one sees that \eqref{F'}, \eqref{Fbounded} and \eqref{efw} are 
satisfied by $F$.

Therefore, we may apply Theorem \ref{thm 3} to conclude 
	$$ \sum_{j\in\Z} \| e^{F_j} P_j E_I(H)\|^2  < \infty .$$
 Finally, notice that $F_j(r) \geq \frac{\ve  }{2} \: |j|^\zeta$ whenever $r\in (0,\frac{\ve}{2} \: |j|^{\zeta/\sigma_+})$. Hence, 
\begin{equation*}
 e^{F_j} \geq  \mathds{1}_{(0, \frac{\ve}{2} |j|^{\zeta /\sigma_+} )} e^{F_j}\geq \mathds{1}_{(0, \frac{\ve}{2}  |j|^{\zeta /\sigma_+} )} e^{\frac{\ve }{2} |j|^\zeta}, 
\end{equation*}
for all $j\in\Z$ which yields \eqref{lemma 1.1 a}. 
\smallskip 
	
\noindent    
Proof of  \eqref{lemma 1.1 b}:   We consider another sequence of functions defined, for any $j\in \Z$,  by
\begin{equation}
G_j(r) =c[ \, r^{ \zeta{\sigma_-}}   - \eta^{\zeta{\sigma_-}} \, \( 1+|j|   \)^\zeta  ]_+\,, \quad r>0\,,
\end{equation}
for  constants $c>0$ and  $\eta\ge\eta_E$ to be fixed below, 
and $0<\zeta\le 1$. 
Using Lemma \ref{cfr}.\eqref{ec2} we have 
\begin{equation}
\abs{G^\prime_j}^2 =\left(c\zeta{\sigma_-}\right)^2\, r^{2( \zeta{\sigma_-}-1)}\mathds{1}_{ ( \eta \( 1+|j|   \)^{1/{\sigma_-}}, \infty)}
\le (V_j-E) \mathds{1}_{ ( \eta \( 1+|j|   \)^{1/{\sigma_-}}, \infty)}\,,
\end{equation}
provided $c \sigma_-\le \lambda_-$. Thus, in view of Remark \ref{obs} this choice of $G_j$ satisfies the requirements 
 \eqref{F'} and \eqref{Fbounded}.
 
 In a similar but easier fashion as above for $F$, one can check 
 \begin{align*}
 	\abs{G_j-G_k}\le c\eta^{\zeta{\sigma_-}}\abs{j-k}^\zeta, 
 \end{align*}
for any $j,k\in \Z$. Hence, for the choice 
 $\eta=\eta_E $ and
  $c=\min\{ \lambda_-/\sigma_-, (a/(2\eta_E))^{1/(\zeta\sigma_-)}\}$, 
  one sees  that \eqref{F'}--\eqref{efw} are 
  satisfied by $G_j$. 
  This implies $ \sum_{j\in\Z} \| e^{G_j} P_j E_I(H)\|^2  < \infty $.

Finally, we note that for $ r\ge  \eta \,[2( 1+|j|)]^{1/{\sigma_-}} $ 
we have  $G_j(r)\ge \frac{c_2}{2} r^{\zeta{\sigma_-}}$.  
Therefore, 
\begin{equation*} 
	e^{G_j(r)} 
	  \geq  
	    \mathds{1}_{ ( \eta \,[2( 1+|j|)]^{\zeta/{\sigma_-}} ,\infty )} e^{G_j(r)} 
	  \geq
        \mathds{1}_{ ( \eta \,(4|j|)^{\zeta/{\sigma_-}} ,\infty )} \, e^{\frac{c_2}{2} r^{\zeta{\sigma_-}}}\,. 
\end{equation*}
This concludes the proof of \eqref{lemma 1.1 b}.

\end{proof}

%
%
%
%

%%%%%%%%%%%%%%%%%%%%%%%%%%%%%

 %%%%%%%%%%%%%%%%%%%%%%%%%%%%%%%%%%%%%%%%%%%%%%%%%%%%%%%%%%%%
 \section{A remark on a model with  mobility edge}\label{med}
 %%%%%%%%%%%%%%%%%%%%%%%%%%%%%%%%%%%%%%%%%%%%%%%%%%%%%%%%%%%%
 So far in the article, nothing has been said about the limiting case $\sigma_- = \sigma_+  = 1$. On this subsection, we give results on the localization of particles moving under such magnetic fields when no electric field is  present. More precisely, we consider the situation in which the magnetic flux is given by
 \begin{equation}\label{mobility flux}
 \Phi(r) = \lambda \, r \, , \qquad r>0 \, 
 \end{equation}
 for some $\lambda>0,$ with the Hamiltonian $H_0$ being defined through \eqref{eq:22}. The spectral quality of this operator has already been determined, see \cite{SimonMiller} or {\cite[Theorem 6.2]{CFKS}}. In particular, it is proven that $\sigma (H_0) = [0,\infty)$ and the spectrum is dense pure point in $[0,\lambda^2)$ and absolutely continuous in $(\lambda^2,\infty)$.
 For energies above $\lambda^2$ one may use the absolute continuity of the spectrum, similarly as it was first done in \cite{Last1996} for Schr\"odinger operators and then adapted in \cite{MSBallistic2015} to Dirac particles, with rotational symmetry to show ballistic dynamics. The long time dynamics for high energies is therefore understood; we now settle the question about dynamics for low, positive energies.
 
 % \begin{equation}
 % \sup_{t\geq 0} \langle \varphi(t), |\bx|^m  \varphi(t) \rangle = \infty 
 % \end{equation}
 
 First, note that the rotational symmetry of $H_0$ makes the dynamics of $J$ trivial. Therefore, to obtain an estimate on $|\textbf{x}(t)|$ it suffices to adapt Theorem \ref{thm 1} to the present case.  One may go through its proof and realize that the exterior tunnelling estimate \eqref{lemma 1.1 b} is all that is needed. We state both of these adapted results in the following.
 
 \begin{theorem}\label{thm6}
 	Let $H_0$ be the Hamiltonian associated to the quadratic form \eqref{eq:22} with $A\equiv\lambda>0 $, as given in \eqref{rotgauge}. Let $E\in (0,\lambda^2) $ and $I=[0,E]$. Then,  
 	there exist constants $c_- \in (1,\infty)$ and $\delta>0$ such that 
 	\begin{align} 
 	%	\label{lemma 1.1 b}
 	\sum_{j\in\Z}  \| \mathds{1}_{ [c_- |j|  ,\infty ) }(\ab{\bx}) \, e^{\delta |\t{\textbf{x}}|  } \, P_j E_I(H_0)\|^2 < \infty\,.
 	\end{align}
 	Consequently, for every $\nu>0$ there exists a constant $C>0$ such
 	that for all $\vp\in E_{I}(H_0) L^2(\R^2)$
 	\begin{equation}
 	\SP{\varphi(t)}{|\bx|^{\nu }\vp(t)} \leq C \( \n{\vp}^2  +    \SP{\vp }{|J|^{\nu } \vp } \)\,, \quad \textnormal{for}\quad t\geq 0\,,
 	\end{equation}
 	holds.
 \end{theorem}
 
 We conclude that dynamical localization holds for energies $E\in(0,\lambda^2)$, provided the initial data is sufficiently regular in the angular variable, that is,
 \begin{align*}
 {\rm sup}_{t\ge 0} \SP{\varphi(t)}{|\bx|^{\nu }\vp(t)}<\infty\,.
 \end{align*}
 for all $\vp\in E_{I}(H_0) L^2(\R^2)\cap D (|J|^{\nu/2})$.

 \begin{proof}[Proof of Theorem 6]
 	We adapt the argument used to prove Theorem \ref{thm-tunnelling}, i.e. we construct an explicit sequence of functions satisfying \eqref{F'} and \eqref{Fbounded} and apply Theorem \ref{thm 3}. Since we assume $W=0$, we can omit \eqref{efw} and take $\epsilon_0 = c_0 = 0$ during the proof. Let $\delta_0 >0$ and $E \equiv \tilde{E}$. First, notice that the classically allowed regions are simplified to
 	\begin{equation}\label{Cj 0}
 	C_j(E) =    
 	\begin{cases} 
 	[ \, \frac{j}{\lambda  + E^{1/2}}  , \frac{j}{ \lambda - E^{1/2} } \, ] &  j>0   , \\
 	\emptyset &  j \leq 0   . \\
 	\end{cases}
 	\end{equation}
 	Note that the proof breaks down for $E>\lambda^2$ since then this structure is lost; one has in turn $C_j(E) = [ \frac{j}{\lambda  + E^{1/2}} , \infty)$ for $j>0$. Now, pick $\eta_1 = \eta_1 (\lambda, E)>\frac{1}{\lambda^2 - E}$ big enough such that
 	\begin{equation}\label{eta 1}
 	\frac{2\lambda}{\eta_1} \leq \frac{1}{2} ( \lambda^2 - E )
 	\end{equation}
 	holds. Then, let $\delta_1 = \delta_1 (\lambda, E, \eta_1) <\min\{\sqrt{\frac{\lambda^2 - E}{2}}, \frac{a}{2\eta_1} \} $ and define $H_j(r) = \delta_1 (r - \eta_1 |j|)_+$. We estimate for $r\in (\eta_1 |j|,\infty)$
 	$$ V_j(r) -E  = \lambda^2 \(1 - \frac{j}{\lambda r}\)^2 -E \geq \lambda^2 \( 1   -  \frac{2}{\lambda \eta_1} \) - E =(\lambda^2 - E) - \frac{2\lambda}{\eta_1} \geq \frac{\lambda^2 - E}{2} $$ 
 	where the last inequality follows from \eqref{eta 1}. Since $|H_j'|^2 = \delta_1^2\,  \mathds{1}_{(\eta_1 |j|,\infty)} \leq \frac{\lambda^2 - E}{2} \mathds{1}_{(\eta_1 |j|,\infty)}$ we have that \eqref{F'} is satisfied. \eqref{Fbounded} is fulfilled in view of $\eta_1 > \frac{1}{\lambda^2 - E}$ and \eqref{Cj 0}. We finish the proof using Theorem \ref{thm 3}, putting $c_- = 2 \eta_1$ together with $\delta = \min \{ \delta_1/2, \delta_0 \} $ and arguing as in the end of the proof of Theorem \ref{thm-tunnelling}.
 \end{proof}

%%%%%%%%%%%%%%%%%%%%%%%%%%%%%%%%%%%%%%%%%%%%%%%%%%%%%%%%%% 
%%%%%%%%        APPENDIX        %%%%%%%%%%%%%%%%%%%%%%%%%%
%%%%%%%%%%%%%%%%%%%%%%%%%%%%%%%%%%%%%%%%%%%%%%%%%%%%%%%%%%
 \appendix

%%%%%%%%%%%%%%%%%%%%%%%%%%%%%%%%%%%%%%%%%%%%%%%%%%%%%%%%%%
\section{The magnetic Schr\"odinger operator $H_0$}\label{app-polar coordinates}
%%%%%%%%%%%%%%%%%%%%%%%%%%%%%%%%%%%%%%%%%%%%%%%%%%%%%%%%%%
Here we want to review the form definition of the magnetic Schr\"odinger 
operator $H_0$. We have to be a little bit careful, since we want to be able to handle rotationally symmetric, but possibly singular, magnetic fields $B$. 
Recall that we choose the vector potential $A$ in the Poincar\'e gauge given by \eqref{rotgauge}. 
\begin{lemma} \label{lem-square-integrable}
If the magnetic field $B:\R^2 \to \R $ is rotationally symmetric and locally square integrable, then the function  
\begin{equation*}
	\R^2\ni x\mapsto \Phi(|x|)/|x|= \frac{1}{|x|}\int_0^{|x|} B(s) s ds
\end{equation*}
is locally square integrable. In particular, the magnetic vector potential $A$ 
given by  \eqref{rotgauge}  is in $ L^2_{\rm loc} (\R^2) $. 	
\end{lemma}
\begin{proof} This is s simple consequence of Jensen's inequality. Since 
$\frac{2}{r^2}\int_0^r s ds =1$, Jensen's inequality shows 
\begin{equation*}
		\Big(\frac{2\Phi(r)}{r^2}\Big)^2=\Big(\frac{2}{r^2} \int_0^r B(s) sds\Big)^2\le \frac{2}{r^2} \int_0^r B(s)^2 sds
\end{equation*}
for all $r>0$. Thus 
\begin{align}
	\int_{|x|\le R} \Big( \frac{\Phi(|x|)}{|x|} \Big)^2 \, dx 
	  &= 2\pi \int_0^R \frac{r^2}{4}		\Big(\frac{2\Phi(r)}{r^2}\Big)^2 rdr 
	     \le \pi \int_0^R \int_0^r B(s)^2 s ds\, rdr \\
  &= \frac{\pi}{2} \int_0^R B(s)^2 (R^2-s^2) sds <\infty
\end{align}
for all $R>0$. By the definition \eqref{rotgauge} we have $|A(x)| =\frac{\Phi(|x|)}{|x|}$, so also the magnetic vector potential in the 
Poincar\'e gauge $A$ is locally square integrable.  
\end{proof}

Denote by $p  = - i \nabla$ the usual momentum operator. We will need a representation of the magnetic Schr\"odinger operator 
$(p-A)^2$, when $A$ is in the Poincar\'e gauge and the magnetic field is 
rotationally symmetric. This is well--known, but we want to include singular magnetic fields, so we have to be a bit careful.   
\begin{lemma}\label{lem-representation} 
  The quadratic form $q_0$ of the free magnetic Schr\"odinger operator $(p-A)^2$ is given by 
  \begin{align}\label{eq-representation with L}
  	q_0(\varphi,\varphi)
  	  &= \SP{(p-A)\vp}{(p-A)\vp} 
  	    =  \SP{\partial_r\vp}{\partial_r\vp}+ \SP{\frac{1}{r}(\Phi-L)\vp}{\frac{1}{r}(\Phi-L)\vp}
  \end{align}
 for all $\vp\in \calD(q_0)$. Here  $L= x_1 p_2 -x_2 p_1 $ is the generator of rotations, i.e., the angular momentum operator,  in $L^2(\R^2)$, $r=|x|$, and the radial derivative is given by $\partial_r= \frac{x}{|x|}\cdot\nabla$. 
 In particular, $\calD(q_0)= \calD(\partial_r)\cap \calD(\frac{1}{r}(\Phi-L))$. 
\end{lemma}

Before we show this, we collect one more result, which is needed 
\begin{lemma}\label{lem-P^2 angular momentum decomposition}
	The quadratic form corresponding to the kinetic energy $p^2$ in dimension $d\ge 2$ is given by 
	\begin{equation}\label{eq-representation P^2}
		\SP{p\vp}{p\vp} 
		  = \SP{\partial_r\vp}{\partial_r\vp} 
		     + \sum_{1\le j<k\le d} \SP{\frac{1}{r}L_{j,k}\vp}{\frac{1}{r}L_{j,k}\vp}
	\end{equation}
	for all $\vp\in H^1(\R^d)$, the form domain of $p^2$, 
 where $r=|x|$, $\partial_r=\frac{x}{|x|}\cdot \nabla$ is the radial derivative on $\R^d$ and $L_{j,k}= x_jp_k-x_kp_j$, $1\le j<k\le d$ are the angular momentum generators.  
 
  In particular, $H^1(\R^d)= \calD(\partial_r)\cap \cap_{1\le j<k\le d}\calD(\frac{1}{r}L_{j,k})$.
\end{lemma}

\begin{proof}[Proof of Lemma \ref{lem-representation}]
  We assume that $\vp \in \mathcal{C}_0^\infty(\R^2)$, by density. Then since $A$ is locally square integrable, 
  \begin{align*}
  	\SP{\vp}{H_0\vp}
  	  &= \SP{(p-A)\vp}{(p-A)\vp} \\
  	  &= \SP{p\vp}{p\vp} -2 \re \SP{A_1\vp}{p_1\vp}
  	       -2 \re \SP{A_2\vp}{p_2\vp} +\SP{A\vp}{A\vp} 
  \end{align*} 
   Using the explicit form of vector potential in the gauge \eqref{rotgauge}, one also sees 
   \begin{align*}
   	  \SP{A_1\vp}{p_1\vp} 
   	    = -  \SP{x_2\frac{\Phi}{r^2}\vp}{p_1\vp} 
   	      = -  \SP{\frac{\Phi}{r}\vp}{\frac{1}{r}x_2p_1\vp} 
   \end{align*}
   and similarly 
   \begin{align*}
   	  \SP{A_2\vp}{p_2\vp} 
   	    =  \SP{\frac{\Phi}{r}\vp}{\frac{1}{r}x_1p_2\vp} 
   \end{align*}
where all terms are well defined, since  Lemma \ref{lem-square-integrable} shows that $\frac{\phi(|x|)}{|x|}$ is locally square integrable over $\R^2$. Thus 
  \begin{align*}
  	\SP{A_1\vp}{p_1\vp} + \SP{A_2\vp}{p_2\vp} 
  	  = \SP{\frac{\Phi}{r}\vp}{\frac{1}{r}(x_1p_2 -x_2p_1)\vp} 
  	  =  \SP{\frac{\Phi}{r}\vp}{\frac{1}{r}L\vp} 
  \end{align*}
 with $L= x_1p_2 -x_2p_1$. Since also 
   \begin{align*}
  	\SP{A\vp}{A\vp} =  \SP{\frac{\Phi}{r}\vp}{\frac{\Phi}{r}\vp} 
  \end{align*}
 this yields 
 \begin{align*}
 	 \SP{(p-A)\vp}{(p-A)\vp} 
 	   &=  \SP{p\vp}{p\vp} + 2\re \SP{\frac{\Phi}{r}\vp}{\frac{1}{r}L\vp} 
 	       + \SP{A\vp}{A\vp} \\
 	   &=  \SP{p\vp}{p\vp} + 2\SP{\frac{\Phi}{r}\vp}{\frac{1}{r}L\vp} 
 	       +   \SP{\frac{\Phi}{r}\vp}{\frac{\Phi}{r}\vp} 
 \end{align*}
 since the angular momentum commutes with rotationally symmetric functions, 
 so $\SP{\frac{\Phi}{r}\vp}{\frac{1}{r}L\vp} $ is real.  
 Moreover, by Lemma \ref{lem-P^2 angular momentum decomposition}, and again using that $L$ commutes with multiplication by rotationally symmetric functions, this gives 
 \begin{align*}
   \SP{(p-A)\vp}{(p-A)\vp} 
 	   &=  \SP{\partial_r\vp}{\partial_r\vp} + \SP{\frac{1}{r}L\vp}{\frac{1}{r}L\vp} + 2 \SP{\frac{\Phi}{r}\vp}{\frac{1}{r}L\vp} 
 	       + \SP{\frac{\Phi}{r}\vp}{\frac{\Phi}{r}\vp} \\
 	   &= \SP{\partial_r\vp}{\partial_r\vp} + \SP{\frac{1}{r}(\Phi-L)\vp}{\frac{1}{r}(\Phi-L)\vp} 
 \end{align*}
 which proves  \eqref{eq-representation with L} when $\vp\in \calC^\infty_0(\R^2)$. 
 Since $\calC^\infty_0(\R^2)$ is dense in the domain of $q_0$ \cite{Simon1979}  
 and all terms on the right hand side of \eqref{eq-representation with L} 
 are non--negative, a standard density argument shows that the domain of 
 $\calD(q_0)$ is equal to the intersection of $\calD(\partial_r)$ and 
 $\calD(\tfrac{1}{r}(\Phi-L))$. This proves Lemma \ref{lem-representation}.

\end{proof}

\begin{proof}[Proof of Lemma \ref{lem-P^2 angular momentum decomposition} ]
    We can use the same density argument as above to see that it is enough to assume that $\vp\in \calC_0^\infty(\R^d)$. Then 
  \begin{align*}
  	\sum_{1\le j<k\le d} \SP{L_{j,k}\vp}{L_{j,k}\vp} 
  	  = - \sum_{1\le j<k\le d} \SP{\vp}{(x_j\partial_k-x_k\partial_j)^2\vp}
  \end{align*}
  Moreover, 
  \begin{align*}
  	\sum_{1\le j<k\le d} & (x_j\partial_k-x_k\partial_j)^2 
  	  = \frac{1}{2}\sum_{ j\neq k } (x_j\partial_k-x_k\partial_j)^2 \\
  	  &= \frac{1}{2}\sum_{ j\neq k } 
  	        (x_j\partial_k x_j\partial_k - x_j\partial_k x_k\partial_j- x_k\partial_j x_j\partial_k  + x_k\partial_j x_k\partial_j) \\
  	  &=  \frac{1}{2}\sum_{ 1\le j, k \le d} 
  	        (x_j^2\partial_k^2 +  x_k^2\partial_j^2 - \partial_k x_kx_j\partial_j   -x_j\partial_j
  	          - \partial_j x_j x_k\partial_k  -x_k\partial_k) \\
  	   &\phantom{==}     - \sum_{ j } 
  	        (x_j^2\partial_j^2 - \partial_j x_j^2\partial_j - x_j\partial_j) \\
  	  &= |x|^2 \Delta -(\nabla\cdot x)(x\cdot\nabla)  + 2x\cdot\nabla 
  	    =  |x|^2 \Delta -(x\cdot\nabla)^2  -(d- 2)x\cdot\nabla 
  \end{align*}
  Thus 
  \begin{align*}
  	p^2= -\Delta = -\frac{1}{|x|^2}(x\cdot\nabla)^2  -\frac{(d- 2)}{|x|^2}x\cdot\nabla + \frac{1}{|x|^2}\sum_{1\le j<k\le d} L_{j,k}^2\, ,
  \end{align*}
  that is, as quadratic forms 
  \begin{align}\label{eq-representation 1}
  	\SP{p\vp}{p\vp} 
  	  = -\SP{\vp}{\frac{1}{|x|^2}(x\cdot\nabla)^2\vp} 
  	       -\SP{\vp}{\frac{d-2}{|x|^2}(x\cdot\nabla)\vp} 
  	       + \sum_{1\le j<k\le d}\SP{\vp}{\frac{1}{|x|^2}L_{j,k}^2\vp}
  \end{align}
 at least when $\vp\in\calC_0^\infty(\R^d)$. 
 For such $\vp$ we set 
 \begin{align*}
 	\psi(r,\omega) \coloneqq \vp(r\omega)\, ,
 \end{align*}
 when $r\ge 0$ and $|\omega|=1$. 
 That is, $\vp(x)= \psi(|x|, x/|x|)$. Then clearly, $\partial_r\psi(r,\omega)= \omega\cdot\nabla\vp(r\omega)$, so 
 \begin{align}
 	x\cdot\varphi(x) = r\partial_r\psi(r, \omega)
 \end{align}
 with $r=|x|$ and $\omega=x/|x|\in S^{d-1}$. Then 
 \begin{align}
 	-\SP{\vp}{\frac{1}{|x|^2}(x\cdot\nabla)^2\vp} 
 	  &= - \int_{S^{d-1}} d\omega \int_0^\infty dr r^{d-1} \, \overline{\psi(r,\omega)} r^{-2}(r\partial_r)^2 \psi(r,\omega) 
 \end{align}
 Now for fixed $\omega$, we have 
 \begin{align*}
 	 -\int_0^\infty dr r^{d-1} &\, \overline{\psi(r,\omega)} r^{-2}(r\partial_r)^2 \psi(r,\omega) \\
 	   & =  - \left[ r^{d-1} \, \overline{\psi(r,\omega)} \partial_r\psi(r,\omega) \right]_{r=0}^{r=\infty} 
 	   		+ \int_0^\infty dr \overline{(\partial_r(r^{d-2} \, \psi(r,\omega)))} (r\partial_r) \psi(r,\omega) \\
 	   &=  (d-2)\int_0^\infty dr r^{d-3}\overline{ \psi(r,\omega))} \, r\partial_r \psi(r,\omega) 
 	       +  \int_0^\infty dr r^{d-1}|\partial_r\psi(r,\omega))|^2  \, .
 \end{align*}
 the first term in the second line above vanishes, since $\vp$ has compact support and $d\ge 2$. 
 Thus 
 \begin{align*}
 	-\SP{\vp}{\frac{1}{|x|^2}(x\cdot\nabla)^2\vp} 
 	  &=  (d-2)\SP{\vp}{\frac{1}{|x|^2}(x\cdot\nabla)\vp} 
 	    + \SP{\partial_r\vp}{\partial_r\vp} 
 \end{align*}
 with $\partial_r= \frac{x}{|x|}\cdot\nabla$. Using this in \eqref{eq-representation 1} shows  
 \begin{align*}
 	  	\SP{p\vp}{p\vp} 
  	  &= \SP{\partial_r\vp}{\partial_r\vp} 
  	       + \sum_{1\le j<k\le d}\SP{\vp}{\frac{1}{|x|^2}L_{j,k}^2\vp} \\
  	  &=   \SP{\partial_r\vp}{\partial_r\vp} 
  	       + \sum_{1\le j<k\le d}\SP{\frac{1}{|x|}L_{j,k}\vp}{\frac{1}{|x|}L_{j,k}\vp} 
 \end{align*}
 since $L_{j,k}$ commutes with multiplication with radial functions. This proves 
 \eqref{eq-representation P^2}, at least when $\vp\in\calC_0^\infty(\R^d)$. 
 On the other hand, since all the terms on the right hand side of \eqref{eq-representation P^2} are positive and $\calC_0^\infty(\R^d)$ is dense in the Sobolev space $H^1(\R^d)$, the form domain of $p^2$, it is an easy exercise to show that then $\partial_r\vp, \frac{1}{|x|}L_{j,k}\vp\in L^2(\R^d)$ for all $\vp\in H^1(\R^d)$ and \eqref{eq-representation P^2} holds for all $\vp\in H^1(\R^d)$. 
\end{proof}

To work in polar coordinates, we identify the Hilbert space $L^2(\R^2)$ 
with the Hilbert space  $\H=L^2(\R^+ \times \S^1, r \d r\, \d\theta)$ (see \eqref{eq:8} above) 
with  the scalar product 
\begin{align*}
  \< f,g\>_{\H}\equiv \< f,g\>=\int_{\R^+\times [0,2\pi)} \overline{f(r,\theta)}
  g(r,\theta)\, r\d r  \d\theta\,.
\end{align*}
It is well-known and easy to see that the map $\calU:L^2(\R^2)\to \H$, 
defined first for $\vp\in\calC_0^\infty(\R^2)$ by 
\begin{align}
	\widetilde{\vp}(r,\theta)=(\calU\vp)(r,\theta)
	  \coloneqq \vp(r\cos\theta,r\sin\theta)\, ,
\end{align}
extend to a unitary operator $L^2(\R^2)\to\H$.  One easily checks 
\begin{align}
	-i\partial_\theta \widetilde{\vp}(r,\theta)
	  = ((x_1p_2-x_2p_1)\vp)(r\cos\theta,r\sin\theta) 
	  = (\calU(L_{1,2}\vp))(r,\theta)
\end{align}
so 
\begin{align}
  J\coloneqq -i\frac{\partial}{\partial\theta} = \calU L_{1,2}\calU^*	
\end{align}
is the self--adjoint angular momentum operator in the $r,\theta$ coordinates. 

Since the complex exponentials $e^{ij\theta}$, $m\in\Z$ and $0\le \theta\le 2\pi$, are an orthonormal basis 
for $L^2([0,2\pi])$, which we identify with $L^2(\S^1)$,  one can expand every $\widetilde{\vp}\in \H$ as 
\begin{align}\label{eq-expansion}
	\widetilde{\vp}(r,\theta) = \sum_{j\in\Z} \vp_j(r) e^{ij\theta}
\end{align}
where $(\vp_j)_{m\in\Z}\in L^2(\R_+,rdr)$ and $\|\widetilde{\vp}\|_{\H}^2= \sum_{j\in\Z}\|\vp_j\|_{L^2(\R_,rdr)}^2$. Then 
\begin{align}
	J\widetilde{\vp}(r,\theta) = \sum_{j\in\Z} \vp_j(r) je^{ij\theta}\, .
\end{align}
Thus the family of corresponding eigen--projections  $(P_j)_{j\in\Z} $ 
of the angular momentum operator $J$ given by  
$(P_j\widetilde{\vp})(r,\theta)= \vp_{j}(r)e^{ij\theta}$ 
decomposes the underlying Hilbert space. We will often write $\vp_j= P_j\calU\vp$, when $\vp\in L^2(\R^2)$.  
In these coordinates we have  
\begin{proposition}\label{prop-angular momentum decomposition}
	Let $\vp$ be in the domain of the quadratic form $q_0$ corresponding to 
	$(P-A)^2$, and expand 
	$\widetilde{\vp}=\calU\vp$ as in \eqref{eq-expansion}. Then 
	\begin{align}
		q_0(\varphi,\varphi) 
		  = \sum_{j\in\Z}
		        \Big( \SP{\partial_r\vp_j}{\partial_r\vp_j}_{L^2(\R_+,rdr)}  
		               + \SP{\vp_j}{\frac{1}{r^2}(\Phi(r)-j)^2\vp_j}_{L^2(\R_+,rdr)}
		        \Big) 
	\end{align}
\end{proposition}
So the eigen--spaces corresponding to $P_j$ are invariant subspaces for 
the unperturbed magnetic Schr\"odinger operator with a rotationally symmetric magnetic field, when the magnetic vector potential is in the Poincar\'e gauge \eqref{rotgauge}. 
 
 Because of the above identity, it is convenient to recall the defintion of the {effective potential}, namely,
\begin{equation}
\label{eff pot}
V_j(r)\coloneqq \frac{1}{r^2} (\Phi(r)-j)^2\,. 
\end{equation}
By polarization, Proposition  \ref{prop-angular momentum decomposition} 
shows that when $\vp,\psi$ are in the domain of the form $q_0$ 
corresponding to $(p-A)^2$ and 
$\widetilde{\vp}=\calU\vp, \widetilde{\psi}=\calU\psi$ are expanded as in  \eqref{eq-expansion} then
\begin{equation}\label{eq-q_0 in polar coordinates}
	q_0(\vp,\psi) 
	  = \sum_{j\in\Z} 
	      \Big( 
	         \SP{\partial_r\vp_j}{\partial_r\psi_j}_{L^2(\R_+,rdr)}  
		               + \SP{\vp_j}{V_j\psi_j}_{L^2(\R_+,rdr)}
          \Big)
\end{equation}

We need one more result, concerning the form boundedness of potentials $W$ satisfying Condition \ref{con3} with respect to the radial kinetic energy.  
\begin{lemma}\label{lem radial energy boundedness}
  Assume that $v$ is a rotationally symmetric potential which is 
  form bounded with respect to $p^2$, that is, for any $0<\veps $ there exists $C(\veps)<\infty$ with 
  \begin{align}\label{eq-ass1}
  	|\SP{\vp}{v\vp}|\le \veps \n{\nabla\vp}^2 +C(\veps)\n{\vp}^2
  \end{align}
  for all $\vp\in\calD(p)$. Then also 
  \begin{align}
  	|\SP{\vp}{v\vp}|\le \veps \n{\partial_r\vp}^2 +C(\veps)\n{\vp}^2
  \end{align}
  for all $\vp\in\calD(\partial_r)$, where $\partial_r=\tfrac{x}{|x|}\cdot\nabla$ is the radial derivative.
\end{lemma}
\begin{proof}
  We expand $\wti{\vp}= \calU\vp= \sum_{j\in\Z}\vp_j e_j$, where $\vp_j$ are purely radial functions and $e_j$ are the basis of complex exponentials. Then for a radial potential $v$ we have 
  \begin{align*}
  	\SP{\vp}{v\vp}
  	  &=   \sum_{j\in\Z} \SP{P_j\vp}{v\vp} 
  	    = \sum_{j\in\Z} \SP{P_j^2\vp}{v\vp} 
  	    = \sum_{j\in\Z} \SP{P_j\vp}{vP_j\vp} 
        = \sum_{j\in\Z} \SP{\vp_j}{v\vp_j}_{L^2(\R_+,rdr )}
  \end{align*}
  with the angular momentum projections $P_j$. 
  Lifting each $\vp_j$ back to $L^2(\R^2)$, by considering it to be constant in the angular coordinate, i.e., identifying it as the function $\R^2\ni x\mapsto \vp_j(|x|)$,  we see have by assumption \eqref{eq-ass1} 
  \begin{align*}
  	  	\abs{\SP{\vp}{v\vp}}
  	  &\le \veps\sum_{j\in\Z}  \SP{\nabla\vp_j}{\nabla\vp_j}+   C(\veps)\sum_{j\in\Z} \SP{\vp_j}{\vp_j}
  \end{align*}
  Now, since $\vp_j$ lifted back to $\R^2$ is radial, we have 
  \begin{align*}
  	\SP{\nabla\vp_j}{\nabla\vp_j}= \SP{\partial_r\vp_j}{\partial_r\vp_j}
  	  = \SP{\partial_rP_j\vp}{\partial_rP_j\vp} 
  	  = \SP{P_j\partial_r\vp}{\partial_r\vp} 
  \end{align*}
  since each angular momentum projection $P_j$ commutes with the radial part of the kinetic energy. 
  We also have 
  \begin{align*}
  	 \SP{\vp_j}{\vp_j} =  \SP{P_j\vp}{P_j\vp} = \SP{P_j\vp}{\vp}
  \end{align*}
  so combining the above yields 
  \begin{align*}
  	 \abs{\SP{\vp}{v\vp}}
  	  &\le \veps\sum_{j\in\Z}  \SP{P_j\partial_r\vp}{\partial_r\vp} +   C(\veps) \sum_{j\in\Z} \SP{P_j\vp}{\vp} 
  	     =   \veps \SP{\partial_r\vp}{\partial_r\vp} 
  	         +   C(\veps) \SP{\vp}{\vp} 
  \end{align*}
  which proves the claim. 
\end{proof}
\begin{remark}\label{rem:diamagnetic}
	The above result also shows that any radial potential $v$ which is form bounded with respect to the nonmagnetic kinetic energy is also form bounded with respect to the magnetic kinetic energy with a rotationally symmetric magnetic field, with the same constants, since 
	by Lemma \ref{lem radial energy boundedness} we have 
	\begin{align}
		\abs{\SP{\vp}{v\vp}} 
		  &\le \veps \SP{\partial_r\vp}{\partial_r\vp} 
		       + C(\veps)\SP{\vp}{\vp}
		  \le \veps (\SP{\partial_r\vp}{\partial_r\vp} 
		       + \SP{\vp}{V\vp})
		       + C(\veps)\SP{\vp}{\vp} \\
		  &= \veps q_0(\vp,\vp)
		       + C(\veps)\SP{\vp}{\vp} 
	\end{align}
  since the effective potential $V=(V_j)_{j\in\Z}\ge 0$.  
\end{remark}

%%%%%%%%%%%%%%%%%%%%%%%%%%%%%%%%%%%%%%%%%%%%%%%%%%%%%%%%%%%%%
\section{The exponentially twisted magnetic quadratic  form}
\label{app-twisted}
%%%%%%%%%%%%%%%%%%%%%%%%%%%%%%%%%%%%%%%%%%%%%%%%%%%%%%%%%%%%%
Here we show that the twisted operators $e^F H e^{-F}$, or better 
their quadratic forms, are well behaved  for a large class of weights 
$F$. Moreover, the bounds are \emph{uniformly in} $F$ for which   
$\|F'\|_\infty\coloneqq \sup_{j\in\Z}\sup_{r>0}|F_j'(r)|$ is bounded. 

%
%\begin{equation} \label{cond on F again}
%  \begin{split}
%  		& (F_j')^2\leq (1-\epsilon_0)(V_j - \tilde{E} )\chi_{j}^{\perp}(\tilde{E}),   \t{ a.e on } \R^+ \t{ for all } j\in\Z\,, \\
%	& F_j \,\chi_j (\tilde{E}) = 0, \t{ for all but finitely many } j\in\Z\,, \\
%	& \sup_{r>0}\ab{F_j(r) - F_k(r)}\leq \tfrac{a}{2}\ab{j-k}^\zeta,   \t{
%          for all  }   j,k\in\Z\,, 
%  \end{split}	
%\end{equation}
We denote by $\Upsilon_K$ the class of sequences of functions 
$F= (F_j)_{j\in\Z}$ satisfying  $ \|F\|_\infty<\infty $ and  
$ \|F'\|_\infty\le K$. 
\begin{lemma}\label{lem-twisted q_0}
  For any sequence $F= (F_j)_{j\in \Z}$ with $F\in \Upsilon_K$ we have 
  $e^{\pm F}\calD(q_0)\subset \calD(q_0)$.  
  Moreover, the quadratic form corresponding to $T_F\coloneqq e^F H_0 e^{-F}-H_0$, that is, 
  \begin{align*}
  	\SP{\vp}{T_F\vp}\coloneqq q_0(e^F\vp, e^{-F}\vp) - q_0(\vp,\vp)
  \end{align*}  
  is, uniformly in $F\in\Upsilon_K$,  infinitesimally form bounded with respect to $H_0$.  
\end{lemma}
\begin{proof} 
  We have 
  \begin{align*}
  	\partial_r(e^{\pm F_j}\vp_j)  
  	  &=  e^{\pm F_j}\left( \partial_r\vp_j \pm  F_j'\vp_j)\right) 
  \end{align*}
  and since $F_j$ and $F_j'$ are bounded, this implies 
  $  |\partial_r(e^{\pm F_j}\vp_j) | 
  	  \lesssim   |\partial_r\vp_j| + \vp_j|$. 
Since also $\sqrt{V_j}|e^{\pm F_j}\vp_j|\lesssim \sqrt{V_j}|\vp_j|$ 
one sees that $e^{\pm F}\vp \in \calD(q_0)$ as soon as  
$\vp \in \calD(q_0)$

  As quadratic forms  and using Proposition \ref{prop-angular momentum decomposition} we have 
  \begin{align}
  	\SP{\vp}{T_F\vp} 
  	  &=  q_0(e^F\vp, e^{-F}\vp) - q_0(\vp,\vp) 
  	    = \sum_{j\in\Z} \left( 
  	        \SP{\partial_r(e^{F_j}\vp_j)}{\partial_r(e^{-F_j}\vp_j} 
  	          - \SP{\partial_r\vp_j}{\partial_r\vp_j}
  	     \right) \nonumber \\
  	  &= \sum_{j\in\Z} \left( 
  	        \SP{\partial_r\vp_j+F_j'\vp_j}{\partial_r\vp_j-F_j'\vp_j} 
  	         - \SP{\partial_r\vp_j}{\partial_r\vp_j}
  	     \right) \nonumber \\
  	  &= \sum_{j\in\Z} \left( 
  	        \SP{F_j'\vp_j}{\partial_r\vp_j} -  \SP{\partial_r\vp_j}{F_j'\vp_j} - \SP{F_j'\vp_j}{F_j'\vp_j} 
  	     \right) \label{eq-superdooper2}
  \end{align}
  since $e^{F_j}$ commutes with the effective potential $V_j$ for 
  all $j\in\Z$. Thus, for all $0<\veps\le 1$,  
  \begin{align*}
  	|\SP{\vp}{T_F\vp} |
  	  &\le \sum_{j\in\Z} \left( 
  	          2\n{F_j'\vp_j}\n{\partial_r\vp_j} - \SP{F_j'\vp_j}{F_j'\vp_j}
  	        \right)
  	     \le  \sum_{j\in\Z} \left( 
  	           \veps \n{\partial_r\vp_j}^2 +(\veps^{-1}-1) \n{F_j'\vp_j}^2 
  	         \right) \\
  	  &\le \veps\n{\partial_r \vp}^2 + K(\veps^{-1}-1)\n{\vp}^2 
  	     \le  \veps q_0(\vp,\vp) + K(\veps^{-1}-1)\n{\vp}^2 \, ,
  \end{align*}
  which finishes the proof.
\end{proof}
\begin{remark}
  Using \cite[Theorem VI.1.33]{kato1995perturbation} this implies that, uniformly 
  in $F\in\Upsilon_K$,  
  \begin{align}
  	\SP{\vp}{e^F H_0 e^{-F}\vp} 
  	  \coloneqq q_0(e^F\vp, e^{-F}\vp) = q_0(\vp,\vp) +  \SP{\vp}{T_F\vp} 
  \end{align}
  yields a non--symmetric sectorial closed quadratic form on $\calD(q_0)$. 
%  Moreover, \cite{kato}[Theorem VI.3.9] shows that the resolvent sets of 
%  $H_0$ and  
\end{remark}
To control a perturbation $W$ which is not rotationally symmetric, 
 we recall that the Fourier
 transformation of the angular variable is given through the unitary
 operator $$\F:\H \to \bigoplus_{j\in\Z} L^2(\R^+,\d r)$$ acting as the
 closure of the map
$$\psi\longmapsto (\F\psi)_j\equiv\hat{\psi}_j:=
\(\frac{1}{\sqrt{2\pi}}\int_{0}^{2\pi}
\psi(\,\cdot\,,\theta)e^{-ij\theta}d\theta\)_{j\in\Z}\,,$$
initially defined on $\calU \calC^\infty_0(\R^2)$.

It is easy to check that, for any $j\in \Z$ and $\vp\in \calU \calC^\infty_0(\R^2)$, 
\begin{equation*}
[P_j \psi](r,\theta) = \hat{\psi}_j(r) e_j(\theta)\,,\quad r>0, \theta
\in [0,2\pi)\,,
\end{equation*}
with $e_j(\theta)\coloneqq e^{ij\theta}/\sqrt {2\pi}$, since $P_j :=
\mathbf{1}\otimes |e_j {\rangle}{ \langle} e_j|$ on $L^2(\R^+)\otimes
L^2(\S^1)\simeq \H$. Moreover, we have
\begin{align}
  \label{eq:11}
  \sps{P_j\varphi}{W P_k\psi}_{\mathcal{H}}=\sps{ \hat{\varphi}_j}{
  \widehat{W}(\cdot, j-k)
   \hat{\psi}_k}_{L^2(\R^+)}\,.
\end{align}

\begin{lemma}\label{w-and-ef}
Let $F= (F_j)_{j\in\Z}$ be a sequence of bounded functions
satisfying \eqref{efw} for some $a>0$ and $0<\zeta\le 1$. Then, for any 
$\varphi\in \calC^\infty_0(\R^2)$,
\begin{align}\label{ffbb}
  \ab{\sps{\varphi}{e^{F} W e^{-F}\varphi}}\le \xi(a,\zeta)\sps{\varphi}{v \,\varphi}\,. 
\end{align}
Moreover,  for any $\varphi\in
C^\infty_0(\R^2)$,
\begin{align}\label{relbound}
  \ab{\sps{\varphi}{ W \varphi}}\le \xi(2a,\zeta)\sps{\varphi}{v \,\varphi}\,. 
\end{align}
Here $v$ is defined through Condition \ref{con3} and $\xi(a,\zeta):=\sum_{k\in\Z} e^{\frac{a}{2}|k|^\zeta}$. 
\end{lemma}
\begin{proof}
We estimate using \eqref{asbe} for any $\varphi\in
C^\infty_0(\R^2)$
\begin{align*}
  \ab{\sps{\varphi}{e^{F} W e^{-F}\varphi}}&\le \sum_{j,k\in\Z}
 \ab{\SP{e^{F_j}P_j\varphi}{W e^{-F_k}P_k\varphi}}\\
&=\sum_{j,k\in\Z} \ab{\SP{e^{F_j} \hat{\vp}_j}{\widehat{W}(\cdot, j-k) e^{-F_k}\hat{\vp}_k}_{L^2(\R^+)}}\\
&\le \sum_{j,k\in\Z} e^{-a\ab{j-k}^\zeta}\SP{e^{F_j} \ab{\hat{\vp}_j}}
{b  \,e^{-F_k}\ab{\hat{\vp}_k}}_{L^2(\R^+)}\\
&\le \sum_{j,k\in\Z} e^{-a\ab{j-k}^\zeta/2}\SP{ \ab{\hat{\vp}_j}}
{b  \,\ab{\hat{\vp}_k}}_{L^2(\R^+)}\\
&\le  \sum_{j,k\in\Z} e^{-a\ab{j-k}^\zeta/2} \norm{b^{1/2}
  \hat{\vp}_j}_{L^2(\R^+)} \norm{b^{1/2}  \hat{\vp}_k}_{L^2(\R^+)}
\end{align*}
where in the last two inequalities we use \eqref{efw} and
Cauchy-Schwarz inequality for the scalar product, respectively. 
We can estimate de last sums applying  Young's inequality for
convolutions to get
\begin{align*}
   \ab{\sps{\varphi}{e^{F} W e^{-F}\varphi}}\le \left( \sum_{k\in\Z} e^{-a\ab{k}^\zeta/2}\right)
\left(  \sum_{j\in\Z}  \norm{b^{1/2}
  {\hat{\vp}_j}}_{L^2(\R^+)}^2\right) =\xi(a,\zeta)\sps{\varphi}{v \,\varphi}\,.
\end{align*}
This proves \eqref{ffbb}. In the case, $F=0$, we clearly obtain the same
estimate  as above  with $a/2$ replaced by $a$. This concludes the
proof of the lemma.  
\end{proof}

\begin{proposition}
  Assume that $W$ satisfies Condition \ref{con3} for some $a>0$, 
  $0<\zeta\le 1$, and  $F= (F_j)_{j\in\Z}\subset PC^1(\R_+,\R)$ 
  functions satisfying \eqref{efw} such that also $\|F\|_\infty, \|F'\|_\infty<\infty$. Then the twisted quadratic form  
  \begin{equation}
  	q_F(\vp,\vp)= q(e^F\vp, e^{-F}\vp)= q_0(e^F \vp, e^{-F}\vp) 
  	+\SP{e^F\vp}{We^{-F}\vp}
  \end{equation}
  is a closed sectorial form  on $\calD(q_0)$. 	Moreover, we have 
	\begin{align}\label{eq-explicit}
		\re\,  q(e^F\vp, e^{-F}\vp) 
		  &= \SP{\partial_r\vp}{\partial_r\vp} + \SP{\vp}{(V-(F')^2)\vp} 
		      + \re \SP{e^F\vp}{We^{-F}\vp}
	\end{align}as quadratic forms on $\calD(q_0)$. 
\end{proposition}
\begin{proof}
	By Lemmas \ref{lem-twisted q_0} and \ref{w-and-ef}, the quadratic forms 
	corresponding to $T_F$ and $e^F W e^{-F}$ are infinitesimally form 
	bounded with respect to $H_0$. Thus we can apply 
	\cite{kato1995perturbation}[Theorem VI.1.33] to the form 
	\begin{align*}
	  q(e^F\vp, e^{-F}\vp) 
	    = q_0(\vp,\vp) + \SP{\vp}{T_F\vp} + \SP{e^F\vp}{We^{-F}\vp}\, . 
	\end{align*}
	to see that is it is closed sectorial form on $\calD(q_0)$. 
	The explicit form  \eqref{eq-expansion} follows from this since by 
	\eqref{eq-superdooper2} we have 
	\begin{align}
		\re \SP{\vp}{T_F\vp} 
          &= \sum_{j\in\Z} \re\left( 
  	        \SP{F_j'\vp_j}{\partial_r\vp_j} -  \SP{\partial_r\vp_j}{F_j'\vp_j} - \SP{F_j'\vp_j}{F_j'\vp_j} 
  	     \right)  \\
  	     &= -\sum_{j\in\Z} \SP{F_j'\vp_j}{F_j'\vp_j} 
  	         = - \SP{F'\vp}{F'\vp}
	\end{align}

\end{proof}

One more result, which we need and recall here, is the (reverse) triangle inequality for $j\mapsto |j|^\zeta$, when $0<\zeta\le 1$. 
\begin{lemma}\label{lem triangle inequality}
	For all $j,k\in\Z$ we have $|j+k|^\zeta\le |j|^\zeta+|k|^\zeta$ and, in particular, also $ \abs{|j|^\zeta-|k|^\zeta}\le |j+k|^\zeta$ 
\end{lemma}
\begin{proof}
  This is well--known, we give the easy argument for the convenience of the reader(s). 
  If $\zeta=1$, this is the usual triangle inequality. So let $0<\zeta<1$ and also $j,k\not=0$. Then 
  \begin{align*}
    |j+k|^\zeta 
      &\le (|j|+|k|)^\zeta = \frac{|j|+|k|}{(|j|+|k|)^{1-\zeta}} 
        = \frac{|j|}{(|j|+|k|)^{1-\zeta}} + \frac{|k|}{(|j|+|k|)^{1-\zeta}} \\
      &\le \frac{|j|}{|j|^{1-\zeta}} + \frac{|k|}{|k|^{1-\zeta}} 
        = |j|^\zeta +|k|^\zeta\, .
  \end{align*}
  and with the usual trick, the reverse triangle inequality 
  $
  	 \abs{|j|^\zeta-|k|^\zeta}\le |j+k|^\zeta
  $
  follows.
\end{proof}

\bigskip
\noindent {\bf Acknowledgments.}
 	D.H. and S.W. gratefully acknowledge the funding by Deutsche Forschungsgemeinschft (DFG) through project ID 258734477--SFB 1173.
 E.S. has been partially funded by Fondecyt (Chile)
 project \# 118--0355. 
 %\bibliographystyle{plain}
 %\bibliography{biblio}

%%%%%
\end{document}